\documentclass{sig-alternate}

\permission{Copyright is held by the International World Wide Web Conference Committee (IW3C2). IW3C2 reserves the right to provide a hyperlink to the author's site if the Material is used in electronic media.}
\conferenceinfo{WWW 2015,}{May 18--22, 2015, Florence, Italy.} 
\copyrightetc{ACM \the\acmcopyr}
\crdata{978-1-4503-3469-3/15/05. \\
http://dx.doi.org/10.1145/2736277.2741640}

\usepackage[ruled, noend, noline, linesnumbered]{algorithm2e}
\usepackage{amsmath,amsfonts,amssymb}
\usepackage[usenames]{color}
\usepackage{tabularx}
\usepackage{paralist}

\usepackage{xspace}

\usepackage{paralist}

\usepackage{graphicx}
\usepackage{grffile}

\usepackage{epstopdf}

\usepackage{boxedminipage}

\usepackage{subfig}

\usepackage[colorlinks,urlcolor=blue,citecolor=blue,linkcolor=blue]{hyperref}

\usepackage{subfig}
\usepackage{tabularx}
\usepackage{array}
\usepackage{subfig}
\usepackage{tabularx}
\usepackage{array}

\usepackage{booktabs}
\usepackage{textcase}

\DeclareCaptionType{copyrightbox}

\def\doctype{2}

\def\tsubmission{1}

\ifnum\doctype=\tsubmission
\newcommand{\full}[1]{}
\newcommand{\submit}[1]{#1}
\else
\newcommand{\full}[1]{#1}
\newcommand{\submit}[1]{}
\fi

\newtheorem{lemma}{Lemma}
\newtheorem{theorem}{Theorem}

\newtheorem{claim}{Claim}
\newtheorem{definition}{Definition}

\newcommand{\meach}{\text{\textnormal{\textbf{each} }}}

\newcommand{\Sec}[1]{\hyperref[sec:#1]{\S\ref*{sec:#1}}} 
\newcommand{\Eqn}[1]{\hyperref[eqn:#1]{(\ref*{eqn:#1})}} 
\newcommand{\Fig}[1]{\hyperref[fig:#1]{Fig.\,\ref*{fig:#1}}} 
\newcommand{\Tab}[1]{\hyperref[tab:#1]{Tab.\,\ref*{tab:#1}}} 
\newcommand{\Thm}[1]{\hyperref[thm:#1]{Thm.\,\ref*{thm:#1}}} 
\newcommand{\Lem}[1]{\hyperref[lem:#1]{Lem.\,\ref*{lem:#1}}} 
\newcommand{\Prop}[1]{\hyperref[prop:#1]{Prop.~\ref*{prop:#1}}} 
\newcommand{\Cor}[1]{\hyperref[cor:#1]{Cor.~\ref*{cor:#1}}} 
\newcommand{\Def}[1]{\hyperref[def:#1]{Defn.~\ref*{def:#1}}} 
\newcommand{\Alg}[1]{\hyperref[alg:#1]{Alg.\,\ref*{alg:#1}}} 
\newcommand{\Ex}[1]{\hyperref[ex:#1]{Ex.~\ref*{ex:#1}}} 
\newcommand{\Clm}[1]{\hyperref[clm:#1]{Claim~\ref*{clm:#1}}} 
\newcommand{\Step}[1]{\hyperref[step:#1]{Step~\ref*{step:#1}}} 
\SetKwComment{tcp}{$\triangleright$ }{}

\newcommand{\cG}{\mathcal{G}}
\newcommand{\cH}{\mathcal{H}}

\newcommand{\cS}{\mathcal{S}}

\newcommand{\cc}{cc}
\newcommand{\ct}{ct}
\newcommand{\degk}{\delta}

\newcommand{\setk}{{\tt set-k}}

\newcommand{\kval}{\kappa}
\newcommand{\rt}{RT}

\newcommand{\fb}{{\tt facebook}\xspace}
\newcommand{\epinion}{{\tt soc-sign-epinions}\xspace}
\newcommand{\notredame}{{\tt web-NotreDame}\xspace}

\newcommand{\wiki}{{\tt wikipedia-200611}\xspace}

\newcommand{\ignore}[1]{}

\begin{document}

\title{Finding the Hierarchy of Dense Subgraphs \\ using Nucleus Decompositions
}

\numberofauthors{4} 
%

\author{
\fontsize{11}{11}\selectfont
        Ahmet Erdem Sar{\i}y\"{u}ce{\small$^{\dag}$}\titlenote{Work done while the author was interning at Sandia National Laboratories, Livermore, CA.}, 
        C. Seshadhri{\small$^{\ddag}$}, 
        Ali P{\i}nar{\small$^{\mathsection}$}, 
        \"{U}mit V. \c{C}ataly\"{u}rek{\small$^{\dag}$}\\%
\fontsize{8}{8}\ttfamily\upshape\selectfont
         sariyuce.1@osu.edu, 
         scomandu@ucsc.edu,
         apinar@sandia.gov,
         umit@bmi.osu.edu\\%
\fontsize{10}{10}\itshape\selectfont
        $^{\dag}$The Ohio State University, Columbus, OH, USA\\
\fontsize{10}{10}\itshape\selectfont
        $^{\ddag}$University of California, Santa Cruz, CA, USA\\
\fontsize{10}{10}\itshape\selectfont
        $^{\mathsection}$Sandia National Labs, Livermore, CA, USA\\
}

\maketitle
\begin{abstract}

Finding dense substructures in a graph is a fundamental graph mining operation, with applications
in bioinformatics, social networks, and visualization to name a few. Yet most standard
formulations of this problem (like clique, quasi-clique, k-densest subgraph) are NP-hard. 
Furthermore, the goal  is rarely to find the ``true optimum", but to identify many (if not all) dense
substructures, understand their distribution in the graph, and ideally determine
relationships among them. Current dense subgraph finding algorithms usually 
optimize some objective, and only find a few such subgraphs without providing any structural
relations.

We define the \emph{nucleus decomposition} of a graph, which represents the graph
as a \emph{forest of nuclei}. Each nucleus is a subgraph where smaller cliques
are present in many larger cliques. The forest of nuclei is a hierarchy by containment,
where the edge density increases as we proceed towards leaf nuclei. Sibling nuclei
can have limited intersections, which enables discovering overlapping dense subgraphs.
With the right parameters, the nucleus decomposition generalizes the classic notions of $k$-cores and $k$-truss decompositions.

We give provably efficient algorithms for nucleus decompositions, and empirically
evaluate their behavior in a variety of real graphs. The tree of nuclei consistently 
gives a global, hierarchical snapshot of dense substructures, and
outputs dense subgraphs of higher quality than other state-of-the-art solutions.
Our algorithm can process graphs with tens of millions of edges in less than an hour. 

\end{abstract}




\section{Introduction} \label{sec:intro}

Graphs are widely used to model relationships in a wide variety of domains 
such as sociology, bioinformatics, infrastructure, the WWW, to name a few.
One of the key observations is that while real-world graphs are often globally
sparse, they are locally dense. In other words, the average degree is often
quite small (say at most 10 in a million vertex graph), but vertex
neighborhoods are often dense. The classic notions of transitivity~\cite{WaFa94} and clustering coefficients~\cite{WaSt98}
measure these densities, and are high for many real-world graphs~\cite{SaCaWiZa10,SePiKo13-tri-j}.

Finding dense subgraphs is a critical aspect of graph mining~\cite{Lee10}.
It has been used for finding communities and spam link farms in web graphs~\cite{Kumar99, Gibson05, Dourisboure07}, graph
visualization~\cite{Alvarez06}, real-time story identification~\cite{Angel12},
DNA motif detection in biological networks~\cite{Fratkin06}, finding correlated
genes~\cite{Zhang05}, epilepsy prediction~\cite{Iasemidis03}, finding price value motifs in financial
data~\cite{Du09}, graph compression~\cite{Buehrer08}, distance query indexing~\cite{Jin09}, and
increasing the throughput of social networking site servers~\cite{Gionis13}.
This is closely related to the classic sociological notion of group cohesion~\cite{BeCo+03,Fo10}.
There are tangential connections to classic community detection, but the objectives are significantly different.
Community definitions involve some relation of inner versus outer connections, while
dense subgraphs purely focus on internal cohesion.

\subsection{The challenges of dense subgraphs} \label{sec:dense}

Our input is a graph $G = (V,E)$. For vertex set $S$, we use $E(S)$ to denote the set of edges internal to $S$.
The \emph{edge density} of $S$ is $\rho(S)=|E(S)|/{|S|\choose 2}$, the fraction
of edges in $S$ with respect to the total possible. The aim is to find
a set $S$ with high density subject to some  size constraint. Typically, we are looking
for large sets of high density. 

In general, one can define numerous formulations that capture
the main problem. The maximum clique problem
is finding the largest $S$ where $\rho(S) = 1$. Finding the densest $S$
of size at least $k$ is the $k$-densest subgraph problem. Quasi-cliques, as defined
recently by Tsourakakis et al.~\cite{Tsourakakis13}, are sets that are almost cliques,
up to some fixed ``defect." Unfortunately, most formulations of finding dense subgraphs are NP-hard, even to approximate~\cite{Hastad96,Fe02,Kh06}.

For graph analysis, one rarely looks for just a single (or the optimal, for whatever notion) dense subgraph. We want to find
many dense subgraphs and understand the relationships among them. Ideally, we would like to see if they nest within each other,
if the dense subgraphs are concentrated in some region, and if they occur at various scales of size and density. 
Our paper is motivated by the following questions.
\begin{asparaitem}
	\item How do we attain a global, hierarchical representation of many dense subgraphs in a real-world graph?
	\item Can we define an efficiently solvable objective that directly provides \emph{many} dense subgraphs? We wish
	to avoid heuristics, as they can be difficult to predict formally.
\end{asparaitem} 

\begin{figure}[t]
  \centering 
	\includegraphics[width=0.4\textwidth,keepaspectratio]{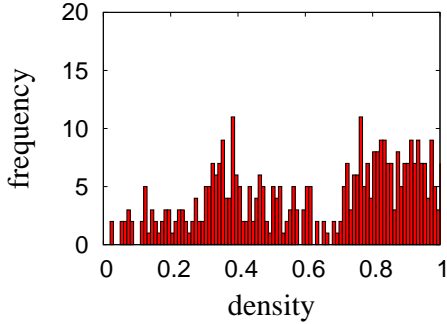}
   \caption{Density histogram of \fb~$(3,4)$-nuclei. $145$ nuclei have density of 
at least $0.8$ and $359$ nuclei are with the density of more than $0.25$.}
   \label{fig:fb-density}
\end{figure}

\begin{figure}[t]
  \centering 
	\includegraphics[width=0.4\textwidth,keepaspectratio]{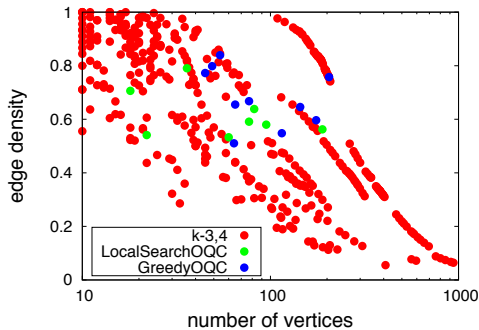}
   \caption{Size vs. density plot for \fb~$(3,4)$-nuclei. $50$ nuclei are larger than 
   $30$ vertices with the density of at least $0.8$.  There are also
   $138$ nuclei larger than $100$ vertices with density of at last $0.25$.}
   \label{fig:fb-comp}
\end{figure}

\begin{figure}[!]
\centering
\begin{minipage}{\linewidth}
\captionsetup{type=subfigure}
 \includegraphics[width=\linewidth]{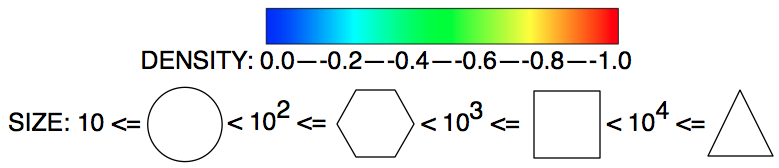}
\end{minipage}
\begin{minipage}{\linewidth}
\captionsetup{type=subfigure}
 \includegraphics[width=\linewidth]{facebook_34.pdf}
\end{minipage}
\caption{$(3,4)$-nuclei forest for \fb. Legends for densities and sizes are shown
at the top. Long chain paths are contracted to single 
edges. In the uncontracted forest, there are 47 leaves and 403 nuclei. Branching depicts the 
different regions in the graph, 13 connected components exist in the top level. Sibling nuclei
have limited overlaps up to 7 vertices.
}
\label{fig:fb-forest}
\end{figure}

\subsection{Our contributions} \label{sec:results}

\textbf{Nucleus decompositions:} Our primary theoretical contribution is the notion
of \emph{nuclei} in a graph. Roughly speaking,
an $(r,s)$-nucleus, for fixed (small) positive integers $r < s$, is a maximal
subgraph where every $r$-clique is part of many $s$-cliques.
(The real definition is more technical and involves some connectivity properties.)
Moreover, nuclei that do not contain one another cannot share an $r$-clique.
This is inspired by and is a generalization of the classic notion of $k$-cores,
and also $k$-trusses (or triangle cores).

We show that the $(r,s)$-nuclei (for any $r < s$) form a hierarchical decomposition
of a graph. The nuclei are progressively denser as we go towards the leaves in the decomposition.
We provide an exact, efficient algorithm that finds all the nuclei and builds the hierarchical
decomposition. In practice, we observe that $(3,4)$-nuclei provide the most interesting
decomposition. We find the $(3,4)$-nuclei for a large variety of
more than 20 graphs. Our algorithm is feasible in practice, and we are able to process
a 39 million edge graph in less than an hour (using commodity hardware).
The source code of our algorithms are available~\footnote{http://bmi.osu.edu/hpc/software/nucleus}.

\textbf{Dense subgraphs from $(3,4)$-nuclei:} The $(3,4)$-nuclei
provide a large set of dense subgraphs for range of densities and sizes.
For example, there are $403$ $(3,4)$-nuclei (of size at least 10 vertices)
in a \fb network of 88K edges. We show the density histogram of these nuclei
in \Fig{fb-density}, plotting the number of nuclei with a given density.
Observe that we get numerous dense subgraphs, and many with density fairly close to $1$. 
In \Fig{fb-comp}, we present a scatter plot of vertex size vs density
of the $(3,4)$-nuclei. Observe that we obtain dense subgraphs over a wide range of sizes.
For comparison, we also plot the output of recent dense subgraph algorithms
from Tsourakakis et al.~\cite{Tsourakakis13}. (These are arguably the state-of-the-art.
More details in next section.) Observe that $(3,4)$-nuclei give dense subgraphs
of comparable quality. In some cases, the output of~\cite{Tsourakakis13} is very close
to a $(3,4)$-nucleus.

\textbf{Representing a graph as forest of $(3,4)$-nuclei:} We build 
the forest of $(3,4)$-nuclei for all graphs experimented on. An example output is that of \Fig{fb-forest}, 
the forest of $(3,4)$-nuclei for the \fb{} network. Each node of the forest is a $(3,4)$-nucleus,
and tree edges indicate containment. More generally, an ancestor nucleus contains
all descendant nuclei. By the properties of $(3,4)$-nuclei, any two incomparable
nodes do not share a triangle. So the branching in the forest represents different
regions of the graph. (All nuclei of less than 10 vertices are omitted.
For presentation, we contract long chain paths in the tree to single edges, so the forest
has less than $403$ nodes.) 

In the nuclei figures, densities are color-coded, with hotter colors indicating higher density. The log of sizes
are coded by shape (circles comprise between 10 and 100 vertices, hexagons between 100 and 1000 vertices, etc.)
For a fixed shape, relative size corresponds to relative size in number of vertices.
We immediately see the hierarchy of dense structures. Observe the colors becoming hotter
as we go towards to leaves, which are mostly red (density $> 0.8$). We see numerous hexagons and large
circles of color between light blue to green. These indicate the larger parent subgraphs of moderate
density (actually density of say $0.25$ is fairly high for a subgraph having many hundreds of vertices).

The branching is also significant, and we can group together the dense subgraphs according to the hierarchy.
We observe such branching in all our experiments, and show more such results later in the paper.
The $(3,4)$-nuclei provide a simple, hierarchical visualization of dense substructures.
They are well-defined and their exact computation is  algorithmically feasible and practical.

We also want to emphasize the overlap between sibling nuclei. While sibling nuclei cannot share triangles, 
they can share edges, thus vertices. We observe roughly 20 pairs of $(3,4)$-nuclei having intersections of 4-6 vertices. 
For larger graphs, we observe many more pairs of intersecting nuclei (with larger intersections).

The rest of the paper is organized as follows: \Sec{prev} summarizes the related work, \Sec{nucleus}
introduces the main definitions and the lemma about the nucleus decomposition, \Sec{generating} gives
the algorithm to generate a nucleus decomposition and provides a complexity analysis, \Sec{experiments} contains the
results of extensive experiments we have, and \Sec{future} concludes the paper by discussing the future directions.

\section{Previous work} \label{sec:prev}

\textbf{Dense subgraph algorithms:} As discussed earlier, most formulations of the densest subgraph problem are NP-hard. Some variants
such as maximum average degree~\cite{Goldberg84, Gallo89} and the recently defined triangle-densest subgraph~\cite{Tsourakakis14} are polynomial time solvable.
Linear time approximation algorithms have been provided by Asashiro et al.~\cite{Asashiro00}, Charikar~\cite{Charikar00},
and Tsourakakis~\cite{Tsourakakis14}.
There are numerous recent practical algorithms for various such objectives: Andersen and Chellapilla's
use of cores for dense subgraphs~\cite{AnCh09}, Rossi et al.'s
heuristic for clique~\cite{Rossi13}, Tsourakakis et al.'s notion of quasi-cliques~\cite{Tsourakakis13}.
These algorithms are extremely efficient and produce excellent output. 
For comparison's sake, we consider Tsourakakis et al.~\cite{Tsourakakis13} as the state-of-the-art, which was compared with previous core-based heuristics and is much superior to prior art. Indeed, their algorithms
are elegant, extremely efficient, and provide high quality output (and much faster
than ours. More discussion in \Sec{runtime}).
These methods are tailored to finding one (or a few) dense subgraphs, and do not give a global/hierarchical view of the structure of dense
subgraphs. We believe it would be worthwhile to relate their methods with our notion of nuclei, to design
even better algorithms.

\textbf{$k$-cores and $k$-trusses:} The concepts of $k$-cores and $k$-trusses form the inspiration for our work.
A $k$-core is a maximal subgraph where each vertex has  minimum degree $k$, while a $k$-truss is a subgraph where
each edge participates in at least $k$ triangles. 
The first definition of $k$-cores was given by Erd\H{o}s and Hajnal~\cite{ErHa66}.
It has been rediscovered numerous times in the context of graph orientations and is alternately called the coloring number and degeneracy~\cite{LiWh70,kcore}.
The first linear time algorithm for computing $k$-cores was given by Matula and Beck~\cite{MaBe83}.
The earliest applications of cores to social networks was given by Seidman~\cite{kcore}, and it 
is now a standard tool in the analysis of massive networks.
The notions of $k$-truss or triangle-cores were independently proposed by 
Cohen~\cite{ktruss}, Zhang and Parthasarathy~\cite{Zhang12}, and Zhao and Tung~\cite{Zhao13}
for finding clusters and for network visualization. They all provide efficient algorithms
for these decompositions, and Cohen~\cite{ktruss} and Wang and Cheng~\cite{WaCh12} explicitly focus on massive scale. 
In~\cite{Wang10}, Wang et al. proposed
DN-graph, a similar concept to $k$-truss, where each edge should be involved in $k$ triangles, and
adding or removing a vertex from DN-graph breaks this constraint.
Apart from the $k$-core and $k$-truss definitions, $k$-plex and $k$-club
subgraph definitions have drawn a lot of interest as well. In a $k$-plex subgraph, each vertex is connected to all but at most  $k-1$ other vertices~\cite{kplex}, which complements the  $k$-core definition. In a $k$-club subgraph, the
shortest path from any vertex to other vertex is not more than $k$~\cite{kclub}.
All these methods find subgraphs of moderate density, and give a global decomposition to visualize a graph.

\section{Nucleus decomposition} \label{sec:nucleus}

Our main theoretical contribution is the notion of nucleus decompositions.
We have an undirected, simple graph $G$.
We use $K_r$ to denote an $r$-clique and start with some technical definitions.

\begin{definition} \label{def:cl-path} Let $r < s$ be positive integers and $\cS$ be a set of $K_s$s in $G$.
\begin{asparaitem}
	\item $K_r(\cS)$ the set of $K_r$s contained in some $S \in \cS$.
	\item The number of $S \in \cS$ containing $R \in K_r(\cS)$ is the \emph{$\cS$-degree} of that $K_r$.
	\item Two $K_r$s $R, R'$ are \emph{$\cS$-connected} if there exists a sequence $R = R_1, R_2, \ldots, R_k = R'$ in $K_r(\cS)$
	such that for each $i$, some $S \in \cS$ contains $R_i \cup R_{i+1}$.
\end{asparaitem}
\end{definition}

These definitions are generalizations of the standard notion of the vertex degree
and connectedness. Indeed, setting $r=1$ and $s=2$ (so $\cS$ is a set of edges) 
yields exactly that. Our main definition is as follows. 

\begin{definition} \label{def:nucleus} Let $k$, $r$, and $s$ be positive integers such that $r < s$.  A 
\emph{$k$-$(r,s)$-nucleus} is a maximal union $\cS$ of $K_s$s such that:
\begin{asparaitem}
	\item The $\cS$-degree of any $R \in K_r(\cS)$ is at least $k$.
	\item Any $R, R' \in K_r(\cS)$ are $\cS$-connected.
\end{asparaitem}
\end{definition}

\begin{figure}[!]
\centering
 \includegraphics[width=0.8\linewidth]{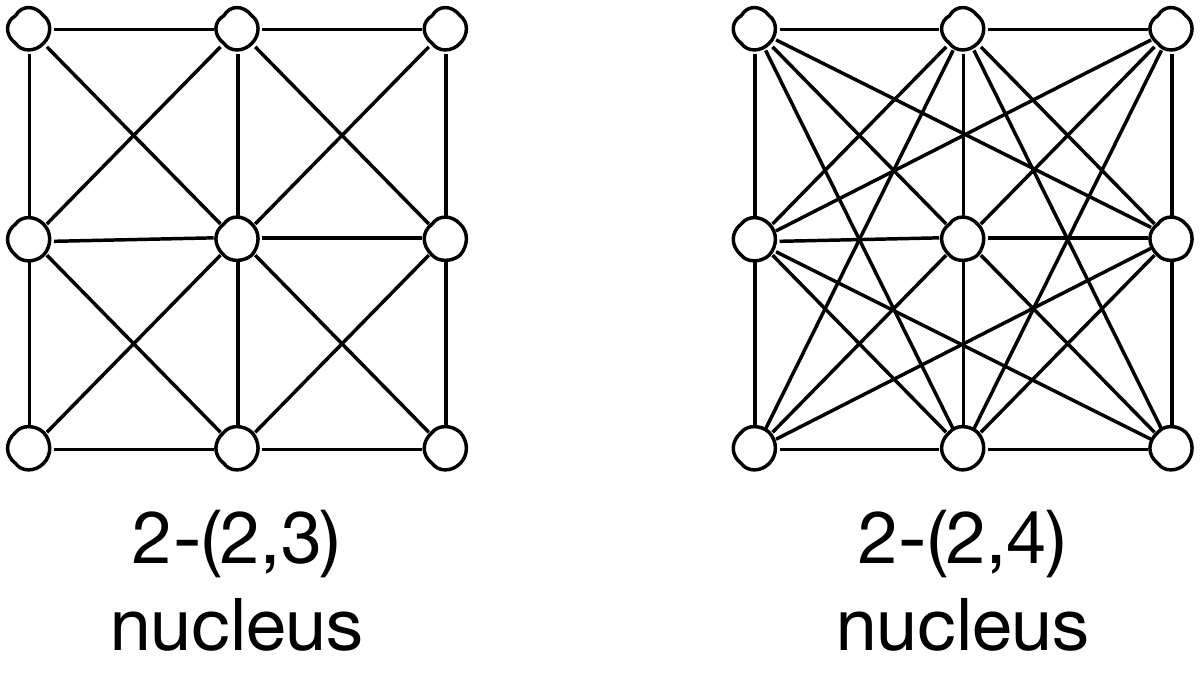}
\caption{Having same number of vertices, $2$-$(2,4)$ nucleus is denser than $2$-$(2,3)$.}
\label{fig:k2x} 
\end{figure}

We simply refer to $(r,s)$-nuclei when $k$ is unspecified. Note that we treat
nuclei as a union of cliques, though eventually, we look at this as a subgraph.
Our theoretical treatment is more convenient in the former setting, and hence
we stick with this definition. In our applications, we simply look at nuclei
as subgraphs.

Intuitively, a nucleus is a tightly connected cluster of cliques. 
For large $k$,
we expect the cliques in $\cS$ to intersect heavily, creating a dense subgraph.
For a fixed $k, r$ and same number of vertices,  the density of the nuclei increases, as we increase $s$.
Consider the example of \Fig{k2x}, where there is a $2$-$(2,3)$-nucleus and a $2$-$(2,4)$-nucleus
on the same number of vertices. Since in the latter case, we need every edge to participate
in at least $2$ $K_4$s, the resulting density is much higher.

\begin{figure}
\centering
 \includegraphics[width=0.8\linewidth]{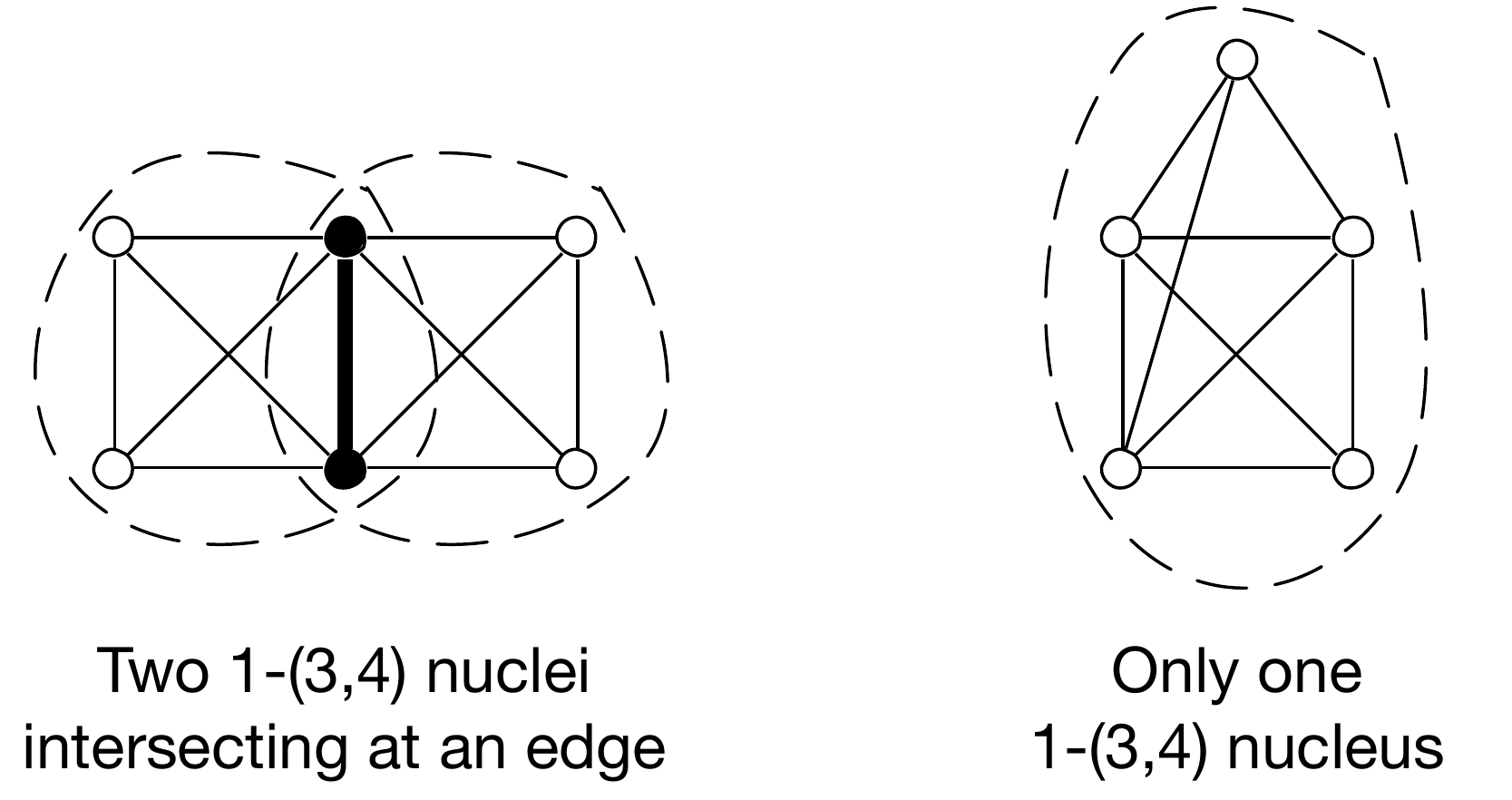}
\caption{The left figure shows two $(3,4)$-nuclei overlapping at an edge.
The right figure has only one $(3,4)$-nucleus}
\label{fig:k34_overlap} 
\end{figure}

As stated earlier, our definitions are inspired by $k$-cores and $k$-trusses.
Set $r=1, s=2$. A $k$-$(1,2)$-nucleus is a maximal (induced) connected subgraph with
minimum vertex degree $k$. \emph{This is exactly a $k$-core.} Setting $r=2, s=3$ gives
maximal subgraphs where every edge participates in at least $k$ triangles, and 
edges are triangle-connected. This is essentially the definition of $k$-trusses or triangle-cores.

So far we only discussed the degree constraint of nuclei. Note that a nucleus is not just
connected in the usual (edge) sense, but requires the stronger property of being $\cS$-connected.
The standard definitions of trusses or triangle-cores omit the triangle-connectedness. 
For us, this is critical. Two cliques of distinct $(r,s)$-nuclei \emph{can} intersect.
For example, when $r > 2$, nuclei can have edge overlaps. This allows for finding even denser
subgraphs, as \Fig{k34_overlap} shows. In the left, cores, trusses, etc. pick up the entire graph.
But there are actually $2$ different $1$-$(3,4)$-nuclei (each $K_4$) intersecting at an edge.
The $(3,4)$-nuclei are denser than the graph itself. Note that any edge disjoint decomposition
would not find two dense subgraphs.

Critically, the set of $(r,s)$-nuclei form a laminar family. A laminar family is a set system
where all pairwise intersections are trivial (either empty or contains one of the sets).

\begin{lemma} \label{sec:laminar} The family of $(r,s)$-nuclei form a laminar family.
\end{lemma}

\begin{proof} Consider $k$-$(r,s)$-nucleus $\cS$ and $k'$-$(r,s)$-nucleus $\cS'$, where $k \leq k'$.
Suppose they had a non-empty intersection, so some $K_s (S)$ is contained in both $\cS$ and $\cS'$.
Observe that $K_r$s in $K_r(\cS)$ are connected to $K_r$s in $K_r(\cS')$. Furthermore,
the $(\cS \cup \cS')$-degree of member of $K_r(\cS \cup \cS')$ is at least $k$.
Hence $\cS \cup \cS'$ satisfies the two conditions of being a nucleus, except maximality.
By $\cS$ is a $k$-$(r,s)$-nucleus, so $\cS \cup \cS' = \cS$. So any non-empty intersection is trivial.
\end{proof}

Consider two nuclei that are not ancestor-descendant.
By the above lemma, these two nuclei (considered as subgraphs of $G$) cannot share a $K_s$.
Actually, the argument above proves that they cannot even share a $K_r$. This is the key
disjointness property of nuclei.

Every laminar family is basically a hierarchical set system. Alternately, every laminar family
can be represented by a forest of containment. For every nucleus $\cS$, any other nucleus intersecting $\cS$
is either contained in $\cS$ or contains $\cS$. Furthermore, all these sets are nested in each other.
It makes sense to talk of the smallest sized nucleus containing $\cS$.
This leads to the main construct we use to represent nuclei. 

\begin{definition} \label{def:tree} Fix $r < s$. Define the \emph{forest of $(r,s)$-nuclei} as follows.
There is a node for each $(r,s)$ nucleus. The parent of every nucleus is the smallest (by cardinality)
other nucleus containing it. 
\end{definition}

In our figures, we will only show the internal nodes of out degree at least $2$, and contract any path of out degree $1$ vertices
into a single path. This preserves all the branching of the forest.

\section{Generating nucleus\\ decompositions} \label{sec:generating}

Our primary algorithmic goal is to construct the tree of nuclei. The algorithm is a direct
adaptation of the classic Matula-Beck result of getting $k$-cores in linear time~\cite{MaBe83}.
There are numerous technicalities involved in generalizing the proof. Intuitively, we do the following.
Construct a graph $\cH$ where the nodes are all $K_r$s of $G$ and there is an edge connecting
two $K_r$s if they are contained in a single $K_s$ of $G$. We then perform a core decomposition
on $\cH$. Actually, this does not work. Edges of $G$ (obviously) contain exactly $2$ vertices
of $G$, and the procedure above exactly produces nuclei for $r=1,s=2$. In general,
a $K_s$ contains $s\choose r$ $K_r$s, and the graph analogy above is incorrect. At some level,
we are performing a hypergraph version of Matula-Beck. The proofs therefore need
to be adapted to this setting.

Analogous to $k$-cores, the main procedure \setk{} (Algorithm~\ref{alg:find-nucleus})
assigns a number, denoted by $\kval(\cdot)$, to each $K_r$ in $G$.

\begin{algorithm}
\caption{\setk($G, r, s$)}
\label{alg:find-nucleus}
  Enumerate all $K_r$s and $K_s$s in $G(V,E)$\;
  For every $K_r$ $R$, initialize $\degk(R)$ to be the number of $K_s$s containing $R$\;
  Mark every $K_r$ as unprocessed\;
  \For{\meach unprocessed $K_r$ $R$ with minimum $\degk(R)$}{
  	$\kval(R) = \degk(R)$\;
    Find set $\cS$ of $K_s$s containing $R$\;
    \For{\meach $S \in \cS$}{
    	\If{any $K_r$ $R' \subset S$ is processed}{
    		Continue\;
    	}
    	\For{\meach $K_r$ $R' \subset S$, $R' \neq R$}{
    		\If{$\degk(R') > \degk(R)$}{ \label{step:11}
    			$\degk(R') = \degk(R') - 1$ \label{step:12}\;
    		}
    	}
    }
  	Mark $R$ as processed\;
  }
  \Return array $\kval(\cdot)$ \;
\end{algorithm}

It is convenient to denote the set of $K_r$s in $G$
by $R_1, R_2, \ldots$, where $R_i$ is the $i$th processed $K_r$ in \setk.
We will refer to this index as \emph{time}. When we say ``at time $t$",
we mean at the beginning of the iteration where $R_t$ is processed.

\begin{claim} \label{clm:mono} The sequence $\{\kval(R_i)\}$ is monotonically non-decreasing.
\end{claim}

\begin{proof} This holds because the loop goes $R$ in non-decreasing order of $d(R)$ and 
\Step{11} ensures that no new value of $\degk(\cdot)$ decreases below the current $\kval(R)$.
\end{proof}

\begin{asparaitem}
	\item Because of \Clm{mono}, we can define \emph{transition time} $t_i$ to be the first time when the $\kval$-value becomes $i$.
Formally, $t_i$ is the unique index such that $\kval(R_{t_i}) = i$ and $\kval(R_{t_i-1}) < i$. 
	\item We say $K_s$ $S$ is \emph{unprocessed at time $t$} if all $R \in K_r(S)$ are unprocessed at time $t$.
	This set of $K_s$s is denoted by $\cS_t$.
	\item The \emph{supergraph $\cG_t$} has node set $K_r(\cS_t)$, and $R, R' \in K_r(\cS_t)$ are connected
	by a link if $R \cup R'$ is contained in some $K_s$ of $\cS_t$. Links are associated with elements of $\cS_t$
	(and there may be multiple links between $R$ and $R'$).
\end{asparaitem}

We prove an auxiliary claim relating the $\degk(\cdot)$ values to $\cS_t$.

\begin{claim} \label{clm:unpro} At time $t$, for any unprocessed $K_r$ $R$, 
$\degk(R)$ is at least the $\cS_t$-degree of $R$. If $t = t_k$ (for some $k$),
then $\degk(R)$ is exactly the $\cS_t$-degree of $R$.
\end{claim}

\begin{proof} Pick unprocessed $R'$. The value of $\degk(R)$ is initially the number
of $K_s$s containing $R'$. It is decremented only in \Step{12}, which happens only
when a processed $K_s$ containing $R'$ is found. (Sometimes, the decrement will
still not happen, because of \Step{11}.) Hence, the value of $\degk(R')$ at time $t$
is at least the number of unprocessed $K_s$s containing $R'$.

Suppose $t = t_k$. For any preceding $\hat{t} < t$, the current $\kval(\cdot)$ value
is always at most $k$. For unprocessed (at time $t$) $R$, $\degk(R) > k$. Hence
the decrement of \Step{12} will always happen, and $\degk(R)$ is exactly 
the $\cS_t$-degree of $R$.
\end{proof}

\begin{claim} \label{clm:contain} Every $k$-$(r,s)$-nucleus is contained in $\cS_{t_k}$.
\end{claim}

\begin{proof} Consider $k$-$(r,s)$-nucleus $\cS$. Take the first $R \in K_r(\cS)$
that is processed. At this time (say $t$), no $S \in \cS$ can be processed.
Hence, $\cS \subseteq \cS_t$. By \Clm{unpro}, $d(R)$ is at least the $\cS_t$-degree of $R$,
which is at least the $\cS$-degree of $R$. The latter is at least $k$, since $\cS$
is a $k$-$(r,s)$-nucleus. By definition of $t_k$, $t \geq t_k$ and hence $\cS_t \subseteq \cS_{t_k}$.
Thus, $\cS \subseteq \cS_{t_k}$.
\end{proof}

The main lemma shows that the output of \setk{} essentially tells us the nuclei.

\begin{lemma} \label{lem:nuclei} The $k$-$(r,s)$-nuclei are exactly the links (which are $K_s$s) of connected
components of $\cG_{t_k}$.
\end{lemma}

\begin{proof} Consider $k$-$(r,s)$-nucleus $\cS$. By \Clm{contain}, it is contained 
in $\cS_{t_k}$. By the nucleus definition, $\cS$ is connected (as links) in $\cG_{t_k}$. Let $\cS'$
be the (set of links) connected component of $\cG_{t_k}$ containing $\cS$. 
By \Clm{unpro}, at time $t_k$, for any $R \in K_r(\cS')$, $\degk(R)$ is exactly
the $\cS_{t_k}$-degree of $R$. Since $\cS'$ is a connected component of $\cG_{t_k}$,
the $\cS_{t_k}$-degree is the $\cS'$-degree, which in turn is at least $k$.
In other words, $\cS'$ satisfies both conditions of being a $k$-$(r	,s)$-nucleus, except
maximality. By maximality of $\cS$, $\cS = \cS'$. 
\end{proof}

\textbf{Building the forest of nuclei:} From \Lem{nuclei}, it is fairly straightforward to get all the nuclei. First run \setk{} to 
get the processing times and the $\kval(\cdot)$ values. We can then get all $t_k$ times as well. Suppose for any $K_r$ in $G$,
we can access all the $K_s$s containing it. Then, it is routine to traverse $\cG_{t_k}$ to get the links of connected
components. To avoid traversing the same component repeatedly, we produce nuclei in reverse order of $k$.
In other words, suppose all connected components of $\cG_{t_{k+1}}$ have been determined. For $\cG_{t_k}$,
it suffices to determine the connected components involving nodes processed in time $[t_k, t_{k+1})$.
Any time a traversal encounters a node in $\cG_{t_{k+1}}$, we need not traverse further. This is because all
other connected nodes of $\cG_{t_{k+1}}$ are already known from previous traversals. We do not get into
the data structure details here, but it suffices to visit all nodes and links of $\cG_0$ exactly once.

\subsection{Bounding the complexity} \label{sec:complex}

There are two options of implementing this algorithm. The first is faster, but has forbiddingly large space.
The latter is slower, but uses less space. In practice, we implement the latter algorithm.
We use $\ct_r(v)$ for the number of $K_r$s containing $v$ and $\ct_r(G)$ for the total
number of $K_r$s in $G$. We denote by
$\rt_r(G)$ the running time of an arbitrary procedure that enumerates all $K_r$s in $G$.

\begin{theorem} \label{thm:fast} It is possible to build the forest of nuclei in $O(\rt_r(G) + \rt_s(G))$ time with $O(\ct_r(G) + \ct_s(G))$ space.
\end{theorem}

\begin{proof} This is the obvious implementation. The very first step of \setk{} requires the
clique enumeration. Suppose we store the global supergraph $\cG = \cG_0$.
This has a node for every $K_r$ in $G$ and a link for every $K_s$ in $G$. The storage is 
$O(\ct_r(S) + \ct_s(G))$. From this point onwards, all remaining operations are linear
in the storage. This is by the analysis of the standard core decomposition algorithm
of Matula and Beck~\cite{MaBe83}. Every time we process a $K_r$, we can delete it and all incident links
from $\cG$. Every link is touched at most a constant number of times during the entire running on \setk.
As explained earlier, we can get all the nuclei by a single traversal of $\cG$.
\end{proof}

\begin{theorem} \label{thm:slow} It is possible to build the forest of nuclei in $O(\rt_r(G) + \sum_v \ct_r(v) d(v)^{s-r})$ time with $O(\ct_r(G))$ space.
\end{theorem}

\begin{proof} Instead of explicitly building $\cG$, we only build adjacency lists when required.
The storage is now only $O(\ct_r(G))$.
In other words, given a $K_r$ $R$, we find all $K_s$s containing $R$ only when $R$ is processed/traversed.
Each $R$ is processed or traversed at most once in \setk{} and the forest building.
Suppose $R$ has vertices $v_1, v_2, \ldots, v_r$. We can find all $K_s$s containing $R$ by looking
at all $(s-r)$-tuples in each of the neighborhoods of $v_i$. (Indeed, it suffices to look at just one
such neighborhood.) This takes time at most $\sum_{R} \sum_{v \in R} d(v)^{s-r} $ 
$= \sum_v \sum_{R \ni v} d(v)^{s-r} $ $= \sum_v \ct_r(v) d(v)^{s-r}$.
\end{proof}

Let us understand these running times. When $r < s \leq 3$, it clearly benefits to go with \Thm{fast}.
Triangle enumeration is a well-studied problem and there exist numerous optimized, parallel solutions
for the problem. In general, the classic triangle enumeration of Chiba and Nishizeki takes $O(m^{3/2})$~\cite{ChNi85}
and is much better in practice~\cite{Co09,ScWa05,SuVa11}. This completely bounds the time and space complexities.

\begin{table*}[!]
\centering
\scriptsize
\renewcommand{\tabcolsep}{2pt}
\begin{tabular}{|l||r|r|l|r|c|l|l|}\hline
		 		& |V|	 	&	 |E| & Description & $\sum_v c_3(v)d(v)$ & $(3,4)$ time (sec) & $\substack{\textrm{~\cite{Tsourakakis13}} \\ \textrm{Density (size)}}$ 
		 		& $\substack{(3,4)\textrm{-nucleus} \\ \textrm{Density (size)}}$ \\\hline \hline
dolphins		 & 	$62$	 & 	$159$ &  Biological   &  $2.2$K	   &  $< 1$	 & 	$0.68 (8)$	 & 	$0.71 (8)$ \\ \hline
polbooks		 & 	$105$	 & 	$441$ &  US Politics Books &   $23.8$K	   &  $< 1$	 & 	$0.67 (13)$	 &    $0.62 (13)$ \\ \hline
adjnoun			 & 	$112$	 & 	$425$ &  Adj. and Nouns &   $17.6$K	   &  $< 1$	 & 	$0.60 (15)$	 & 	$0.22 (32)$ \\ \hline
football		 & 	$115$	 & 	$613$ &  World Soccer 98 &  $26.3$K	   &  $< 1$	 & 	$0.89 (10)$	 & 	$0.89 (10)$	\\ \hline
jazz			 & 	$198$	 & 	$2.74K$	 &  Musicians & 	$2.3$M 	   &  $< 1$	 & 	$1.00 (30)$	 & 	$1.00 (30)$	\\ \hline
celegans n.		 & 	$297$	 & 	$2.34K$	 &  Biological & 	$418$K	   &  $< 1$	 & 	$0.61 (21)$	 & 	$0.91 (10)$ \\ \hline
celegans m.		 & 	$453$	 & 	$2.04K$	 &  Biological & 	$565$K	   &  $< 1$	 & 	$0.67 (17)$	 & 	$0.64 (18)$ \\ \hline
email			 & 	$1.13K$	 & 	$5.45K$	 &  Email & 		$1.2$M	   &  $< 1$	 & 	$1.00 (12)$	 & 	$1.00 (12)$ \\ \hline
facebook		 & 	$4.03K$	 & 	$88.23K$  &  Friendship & 	$712$M	   &  $93$	 & 	$0.83 (54)$	 & 	$0.98 (109)$\\ \hline
protein\_inter.  &  $9.67K$ 	 & 	$37.08K$  &  Protein Inter. &  $35$M  	   &  $< 1$	 & 	$1.00 (11)$	 & 	$1.00 (11)$	\\ \hline
as-22july06		 &  $22.96K$	 & 	$48.43K$  &  Autonomous Sys. &  $199$M 	   &  $< 1$	 & 	$0.58 (12)$	 & 	$1.00 (18)$	\\ \hline
twitter			 & $81.30K$	 & 	$2.68M$  &  Follower-Followee &  $1.8$B	   &  $396$	 & 	$0.85 (83)$	 & 	$1.00 (26)$	\\ \hline
soc-sign-epinions &  $131.82K$  &  $841.37K$  &  Who-trust-whom &  $1.4$B	   &  $242$	 & 	$0.71 (79)$	 & 	$1.00 (112)$ \\ \hline
coAuthorsCiteseer  &  $227.32K$  &  $814.13K$ &  CoAuthorship  &  $2.1$B	   &  $50.1$	 & 	$1.00 (87)$  & 	$1.00 (87)$ \\ \hline
citationCiteseer  &  $268.49K$  &  $1.15M$ &  Citation	 &  $297$M 	   &  $3.4$	 & 	$0.71 (10)$	 & 	$1.00 (13)$	\\ \hline
web-NotreDame	 &  $325.72K$	 &  $1.49M$ &  Web			 &  $33.9$B	   &  $671$	 & 	$1.00 (151)$ & 	$1.00 (155)$ \\ \hline
amazon0601	 &   $403.39K$ &  $3.38M$ &  CoPurchase		 &  $802$M	   &  $23$	 & 	$1.00 (11)$	 & 	$1.00 (11)$ \\ \hline
web-Google	 & 	$875.71K$ & 	$5.10M$	 &  Web			 &  $11.4$B	   &  $163$	 & 	$1.00 (46)$	 & 	$1.00 (33)$	\\ \hline
com-youtube	 & 	$1.13M$	 &  $2.98M$	 &  Social	 &  $451$M  &  $43$	 & 	$0.49 (119)$ & 	$0.92 (24)$	\\ \hline
as-skitter	 & 	$1.69M$	 &  $11.09M$  &  Autonomous Sys.  &  $1.6$B		   &  $1,036$	 & 	$0.53 (319)$ & 	$0.94 (91)$	\\ \hline
wikipedia-2005  &  $1.63M$  &  $19.75M$	 &  Wikipedia Link  &  $741$B	  	   &  $1,312$	 & 	$0.53 (33)$ & 	$0.82 (14)$ \\ \hline
wiki-Talk  & 	$2.39M$	 & 	$5.02M$	 &  Wikipedia User  &  $136$B	  	   &  $605$	 & 	$0.48 (321)$ & 	$0.59 (95)$	\\ \hline
wikipedia-200609  & 	$2.98M$	 &  $37.26M$  &  Wikipedia Link &  $2,015$B  &  $2,830$	 & 	$0.49 (376)$ & 	$0.62 (103)$ \\ \hline
wikipedia-200611  & 	$3.14M$	 &  $39.38M$  &  Wikipedia Link &  $2,197$B  &  $3,039$	 & 	$1.00 (55)$ & 	$1.00 (32)$ \\ \hline								
\end{tabular}
\caption{Important statistics for the real-world graphs of different types and sizes. Largest graph in the dataset has more than $39M$ edges. Times are in seconds. Density of subgraph $S$ is $|E(S)|/{|S|\choose 2}$ where
$E(S)$ is the set of edges internal to $S$. Sizes are in number of vertices.}
\label{tab:properties}
\end{table*}

For our best results, we build the $(3,4)$-nuclei, and the number of $K_4$s is too large to store.  
We go with \Thm{slow}. The storage is now at most the number of triangles, which is manageable.
The running time is basically bounded by $O(\sum_v \ct_r(v) d(v))$. The number of triangles
incident to $v$, $\ct_3(v)$ is $\cc(v) d(v)^2$, where $\cc(v)$ is the clustering coefficient of $v$.
We therefore get a running time of $O(\sum_v \cc(v) d(v)^3)$. This is significantly superlinear,
but clustering coefficients generally decay with degree~\cite{SaCaWiZa10,SePiKo13-tri-j}. Overall,
the implementation can be made to scale to tens of millions of edges with little difficulty.

\section{Experimental Results}\label{sec:experiments}

We applied our algorithms to large variety of graphs, obtained from 
SNAP~\cite{snap} and UF Sparse Matrix
Collection\cite{UF}. The vital statistics of these graphs are given in \Tab{properties}.
All the algorithms in our framework are implemented in C++ and compiled
with \texttt{gcc 4.8.1} at -O2 optimization level. All experiments are
performed on a Linux operating system running on a machine with two Intel Xeon
E5520 2.27 GHz CPUs, with 48GB of RAM.

We computed the $(r,s)$-nuclei for all choices of $r < s \leq 4$, but do not present
all results for space considerations. We mostly observe that
the forest of $(3,4)$-nuclei provides the highest quality output, both in terms
of hierarchy and density. 

As mentioned earlier, we will now treat the nuclei as just induced subgraphs of $G$.
A nucleus can be considered as a set of vertices, and we take all edges among
these vertices (induced subgraph) to attain the subgraph. The \emph{size} of a nucleus
always refers to the number of vertices, unless otherwise specified.
For any set $S$ of vertices,
the density of the induced subgraph is $|E(S)|/{|S|\choose 2}$, where
$E(S)$ is the set of edges internal to $S$. \emph{We ignore any nucleus with less than 10 vertices. 
Such nuclei are not considered in any of our results.}

For brevity, we present detailed results on only 4 graphs (given in \Tab{properties}):
\fb, \epinion, \notredame, and \wiki. This covers
a variety of graphs, and other results are similar.

\subsection{The forest of nuclei} \label{sec:forest}

We were able to construct the forest of $(3,4)$-nuclei for all graphs in \Tab{properties},
but only give the forests for \fb{} (\Fig{fb-forest}), \epinion{} (\Fig{epinion-forest}),
and \notredame{} (\Fig{notredame-forest}). For the \notredame{} figure, we could
not present the entire forest, so we show some trees in the forest that had nice
branching. The density is color coded, from blue (density $0$) to red (density $1$).
The nuclei sizes, in terms of vertices, are coded by shape: circles correspond to  at most $10^2$ vertices,
hexagons in the range $[10^2, 10^3]$, squares in the range $[10^3,10^4]$, and
triangles are anything larger. The relative size of the shape, is the relative
size (in that range) of the set.

\begin{figure*}[!]
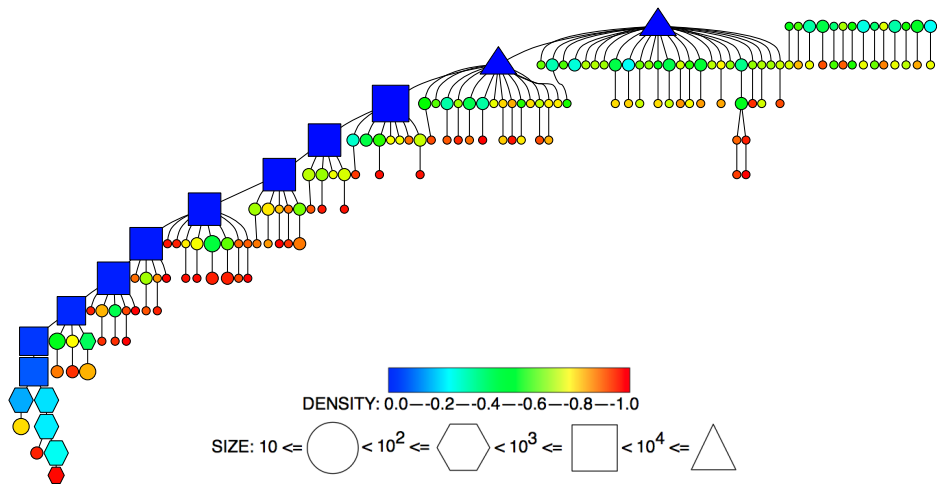

\centering
\begin{minipage}{.7\linewidth}
\captionsetup{type=subfigure}
 \includegraphics[width=\linewidth]{soc-sign-epinions_34_dpi50.pdf}
\vspace*{-15ex}
\end{minipage}
\begin{minipage}{.4\linewidth}
\captionsetup{type=subfigure}
 \includegraphics[width=\linewidth]{legend.png}
\end{minipage}
\caption{$(3,4)$-nuclei forest for \epinion. There are 465 total nodes and 75 leaves in the
forest. There is a clear hierarchical structure of dense subgraphs. Leaves are mostly red
(> $0.8$ density). There are also some light blue hexagons, representing subgraphs of size $\ge 100$
vertices with density of at least $0.2$.}
\label{fig:epinion-forest}
\end{figure*}
\begin{figure*}
\centering
\begin{minipage}{.75\linewidth}
\captionsetup{type=subfigure}
 \includegraphics[width=\linewidth]{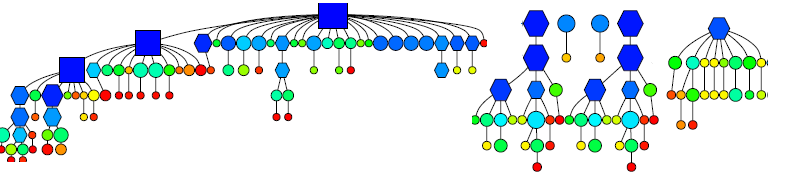}
\end{minipage}
\caption{Part of the $(3,4)$-nuclei forest for \notredame. In the entire forest, there are $2059$ nodes
and $812$ leaves. $79$ of the leaves are clique, up to the size of $155$. There is a nice branching
structure leading to a decent hierarchy.}
\label{fig:notredame-forest}
\end{figure*}

\begin{figure*}[!]
\begin{minipage}{0.3\columnwidth}
\captionsetup{type=subfigure}
\centering
\vspace*{1.5cm}
\includegraphics[height=3.5cm, keepaspectratio]{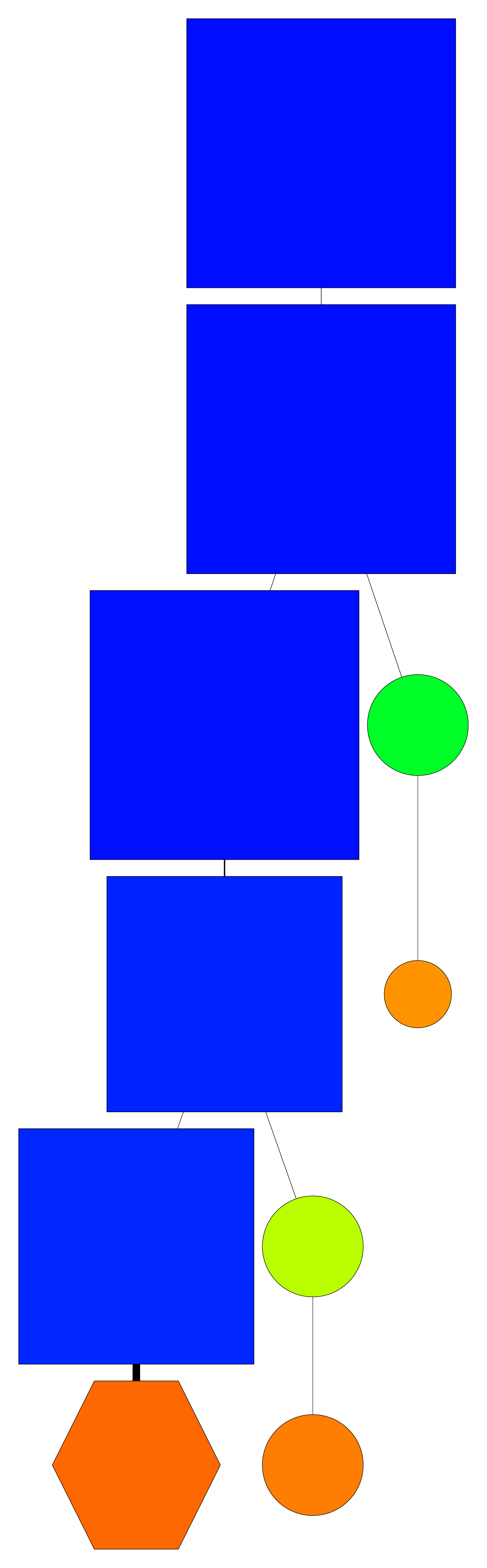}
\vspace*{2cm}
\caption{$(1,2)$-nuclei
}
\label{fig:fb_12}
\end{minipage}
\hspace*{-3ex}
\begin{minipage}{0.3\columnwidth}
\captionsetup{type=subfigure}
\centering
\vspace*{1.5cm}
\includegraphics[height=3.5cm, keepaspectratio]{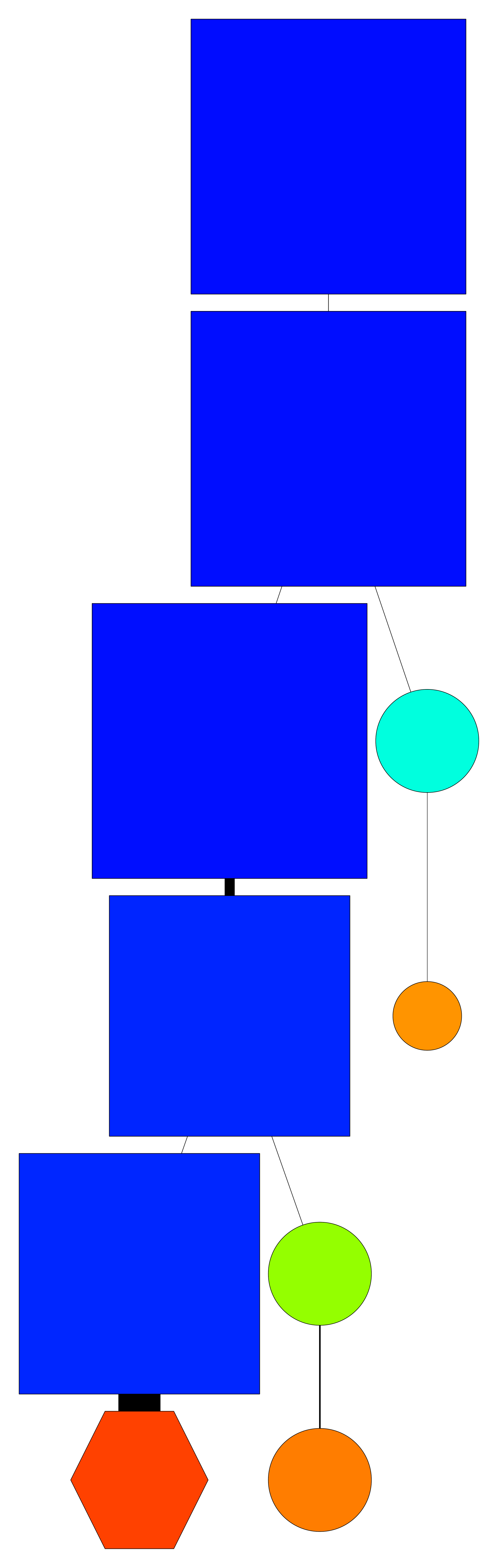}
\vspace*{2cm}
\caption{$(1,3)$-nuclei 
}
\label{fig:fb_13}
\end{minipage}
\hspace*{-3ex}
\begin{minipage}{0.3\columnwidth}
\captionsetup{type=subfigure}
\centering
\vspace*{1.5cm}
\includegraphics[height=3.5cm, keepaspectratio]{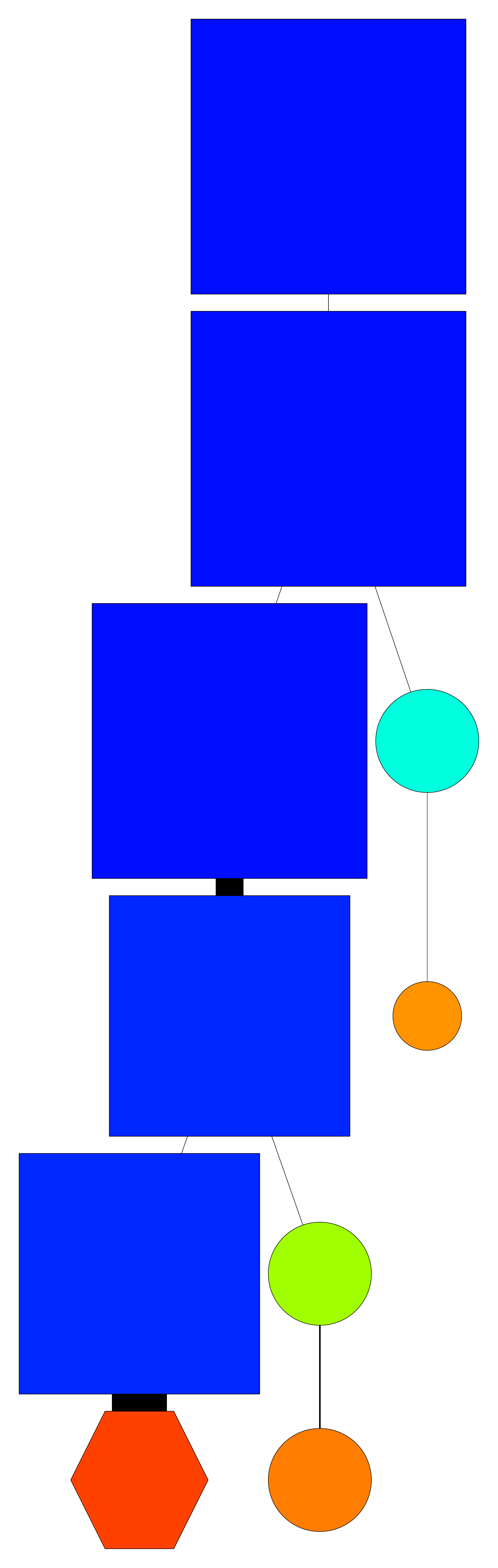}
\vspace*{2cm}
\caption{$(1,4)$-nuclei
}
\label{fig:fb_14}
\end{minipage}
\hspace*{-3ex}
\begin{minipage}{0.65\columnwidth}
\captionsetup{type=subfigure}
\centering
\includegraphics[height=7cm,keepaspectratio]{facebook_23.pdf}
\caption{$(2,3)$-nuclei
}
\label{fig:fb_23}
\end{minipage}
\hspace*{-3ex}
\begin{minipage}{0.65\columnwidth}
\captionsetup{type=subfigure}
\centering
\includegraphics[height=7cm,keepaspectratio]{facebook_24.pdf}
\caption{$(2,4)$-nuclei
}
\label{fig:fb_24}
\end{minipage}
\caption{$(r,s)$-nuclei forests for \fb~ when $r < s \leq 4$ (Except $(3,4)$, which is given in
\Fig{fb-forest}). For $r=1$, trees are more like chains. Increasing $s$ results in
larger number of internal nodes, which are contracted in the illustrations. There is some
hierarchy observed for $r=2$, but it is not as powerful as $(3,4)$-nuclei, i.e., branching
structure is more obvious in $(3,4)$-nuclei.}
\label{fig:tree1}
\end{figure*}

Overall, we see that the $(3,4)$-nuclei provide a hierarchical representation
of the dense subgraphs. The leaves are mostly red, and their densities
are almost always $> 0.8$. But we obtain numerous nuclei of intermediate sizes
and densities. In the \fb{} forest and to some extent in the \notredame{} forest, we see hexagons of light blue to green 
(nuclei of $>100$ vertices of densities of at least $0.2$).
The branching is quite prominent, and the smaller dense nuclei tend to nest
into larger, less dense nuclei. This held in every single $(3,4)$-nucleus forest
we computed. This appears to validate the intuition that real-world networks have
a hierarchical structure. 

\begin{figure*}[!t]
\label{fig:all-comp}
  \centering
  \subfloat[\epinion]{\includegraphics[width=0.32\textwidth,keepaspectratio]{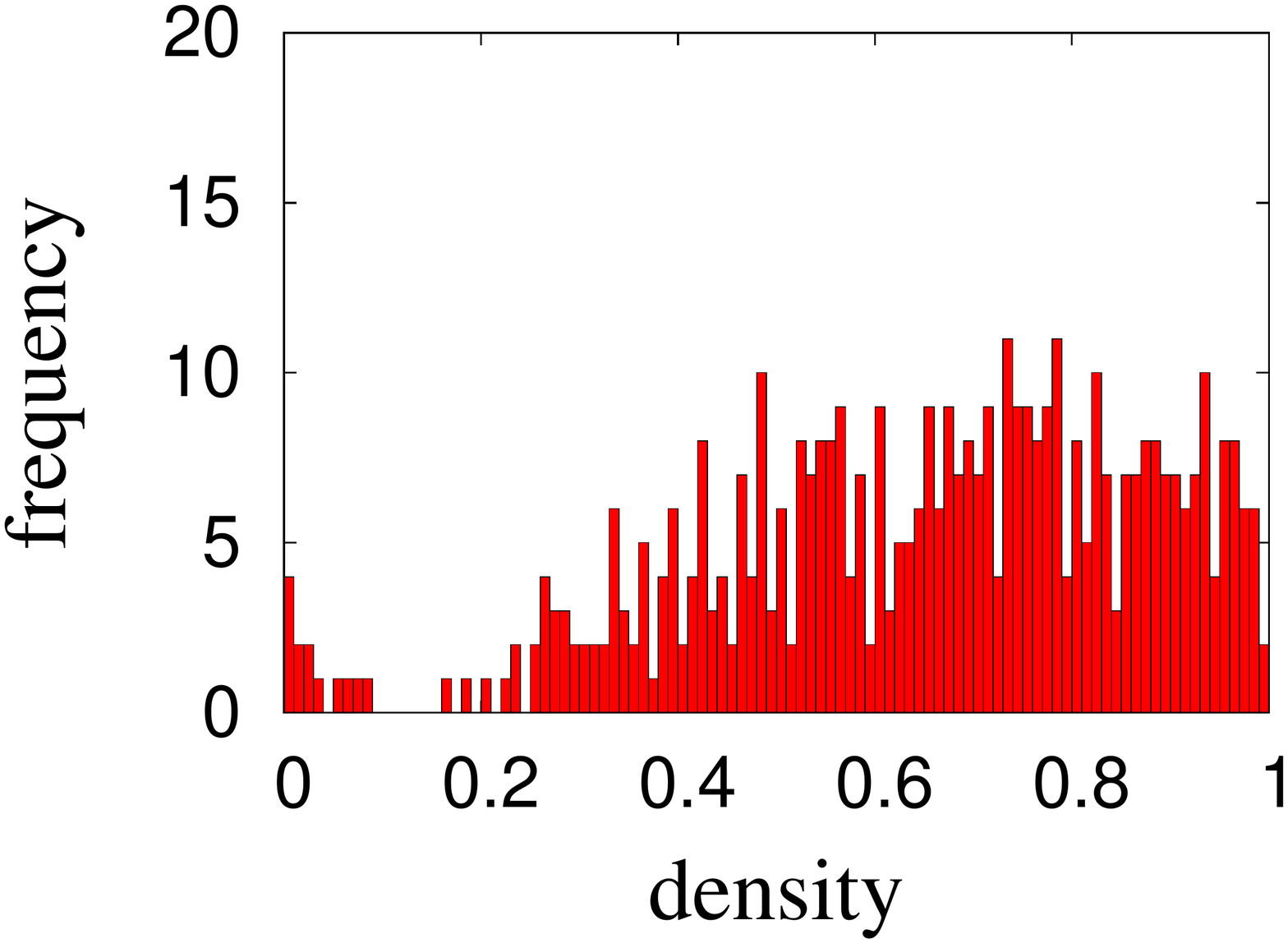}\label{fig:epinion-density} } \ 
  \subfloat[\notredame]{\includegraphics[width=0.32\textwidth,keepaspectratio]{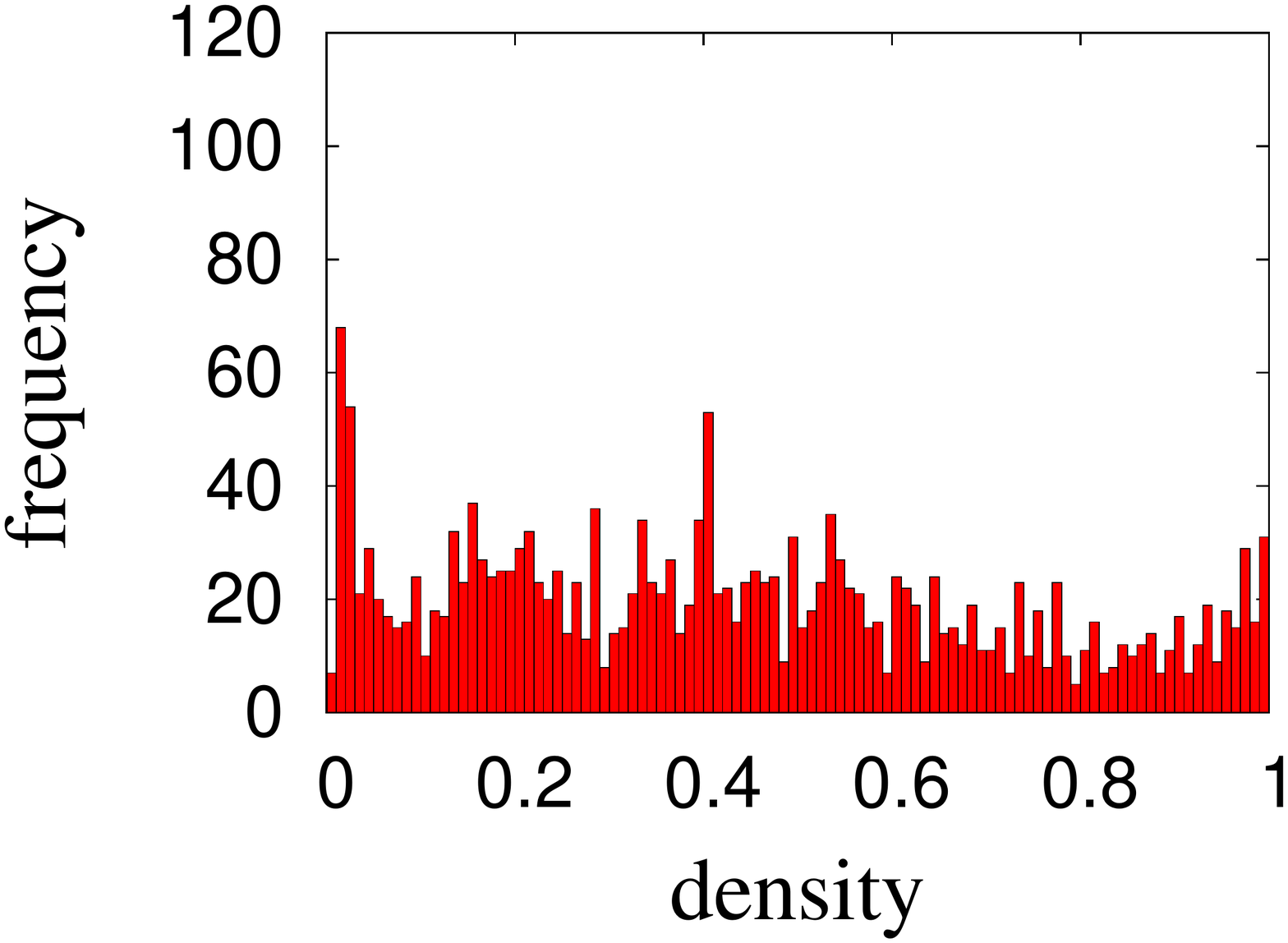}\label{fig:nd-density} } \ 
  \subfloat[\wiki]{\includegraphics[width=0.32\textwidth,keepaspectratio]{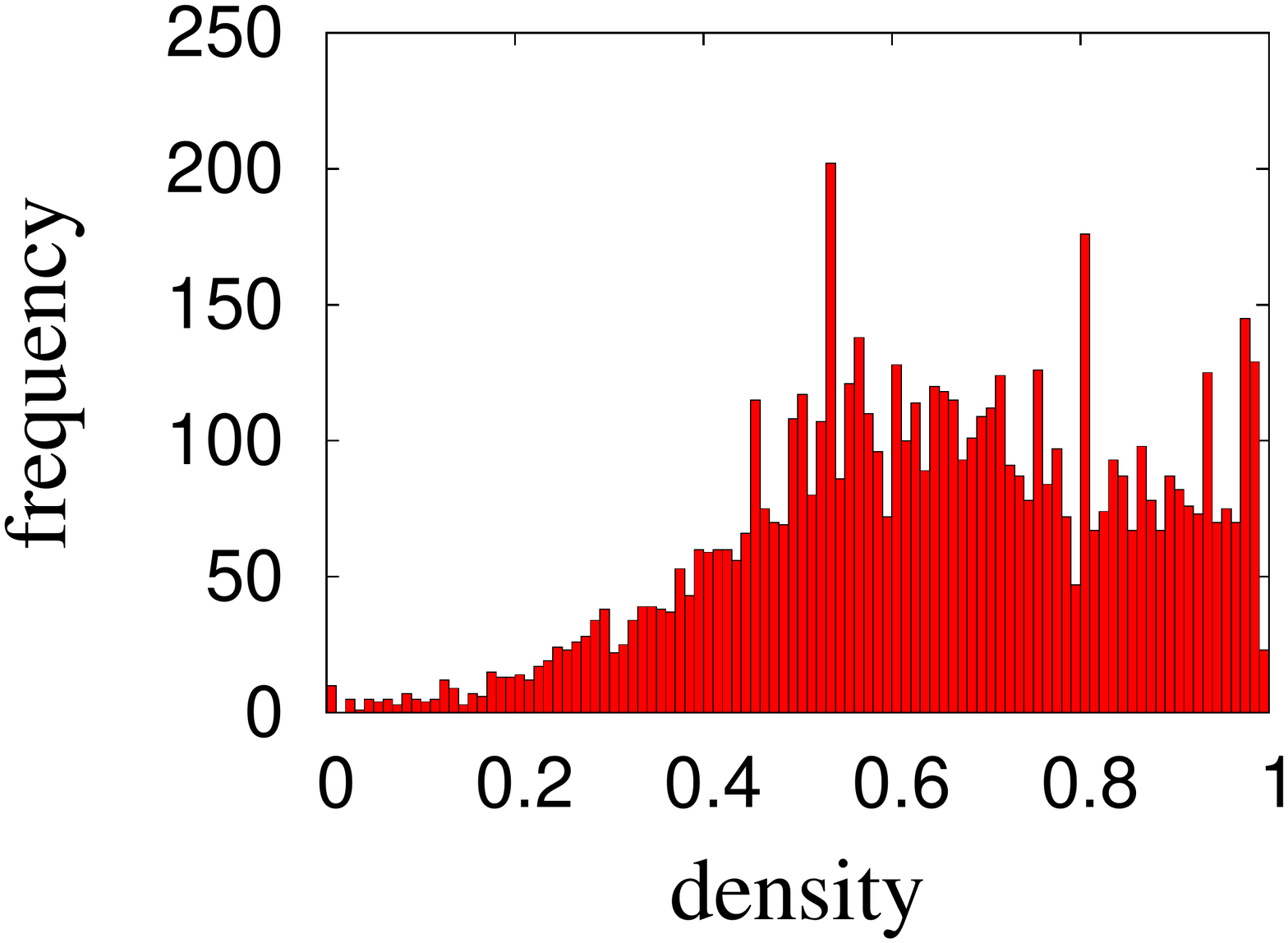}\label{fig:wiki-density} } \ 
   \caption{Density histograms for nuclei of three graphs. $x$-axis (binned) is the density
   and $y$-axis is the number of nuclei (at least 10 vertices) with that density. Number of
   nuclei with the density above 0.8 is significant: $139$ for \epinion, $355$ for \notredame,
   and $1874$ for \wiki. Also notice that, the mass of the histogram is shifted to right in \epinion{}
   and \wiki graphs.\label{fig:density}}
\end{figure*}

\begin{figure*}[!t]
  \centering
  \subfloat[\epinion]{\includegraphics[width=0.32\textwidth,keepaspectratio]{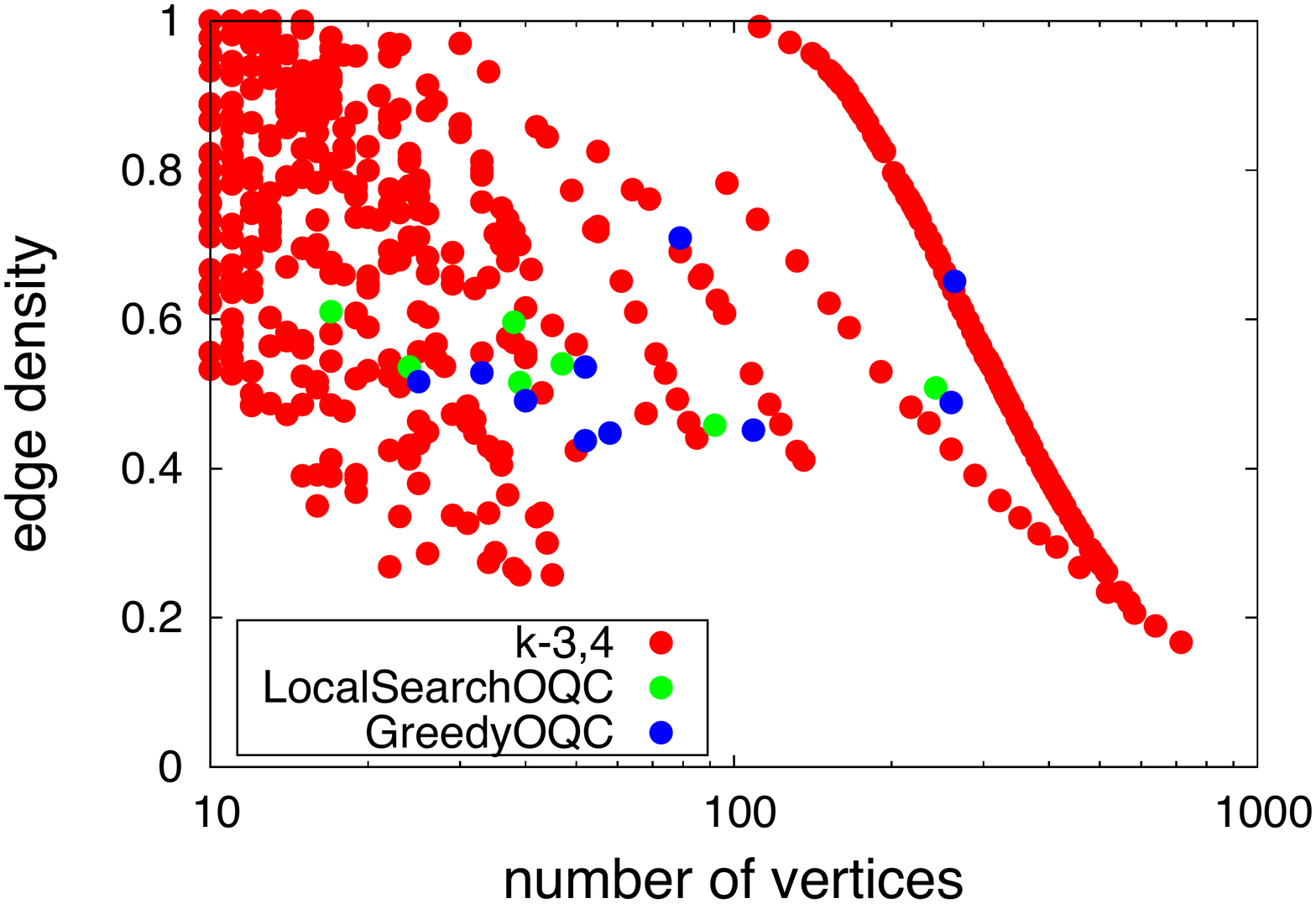}\label{fig:epinion-scatter} } \ 
  \subfloat[\notredame]{\includegraphics[width=0.32\textwidth,keepaspectratio]{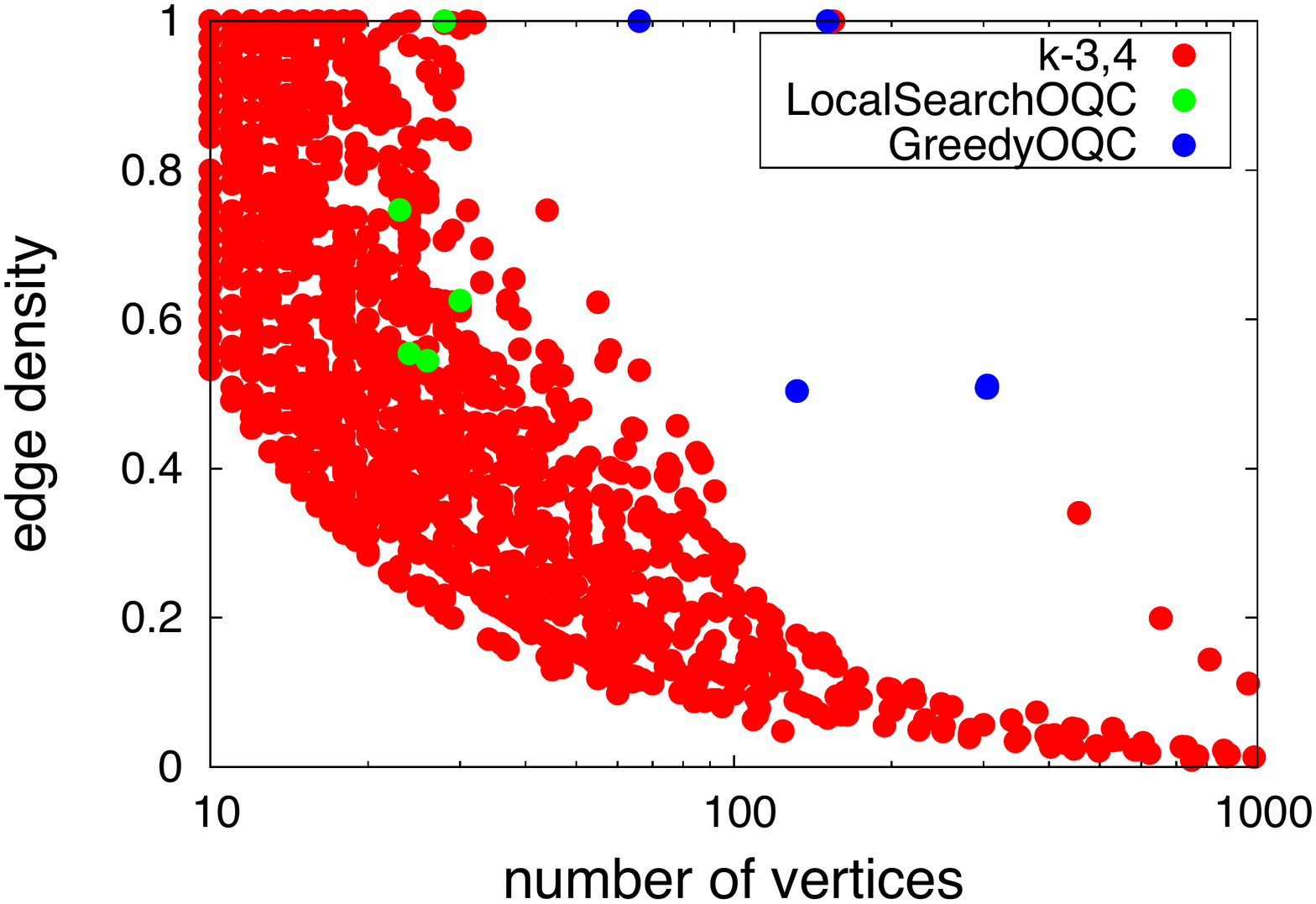}\label{fig:nd-scatter} } \ 
  \subfloat[\wiki]{\includegraphics[width=0.32\textwidth,keepaspectratio]{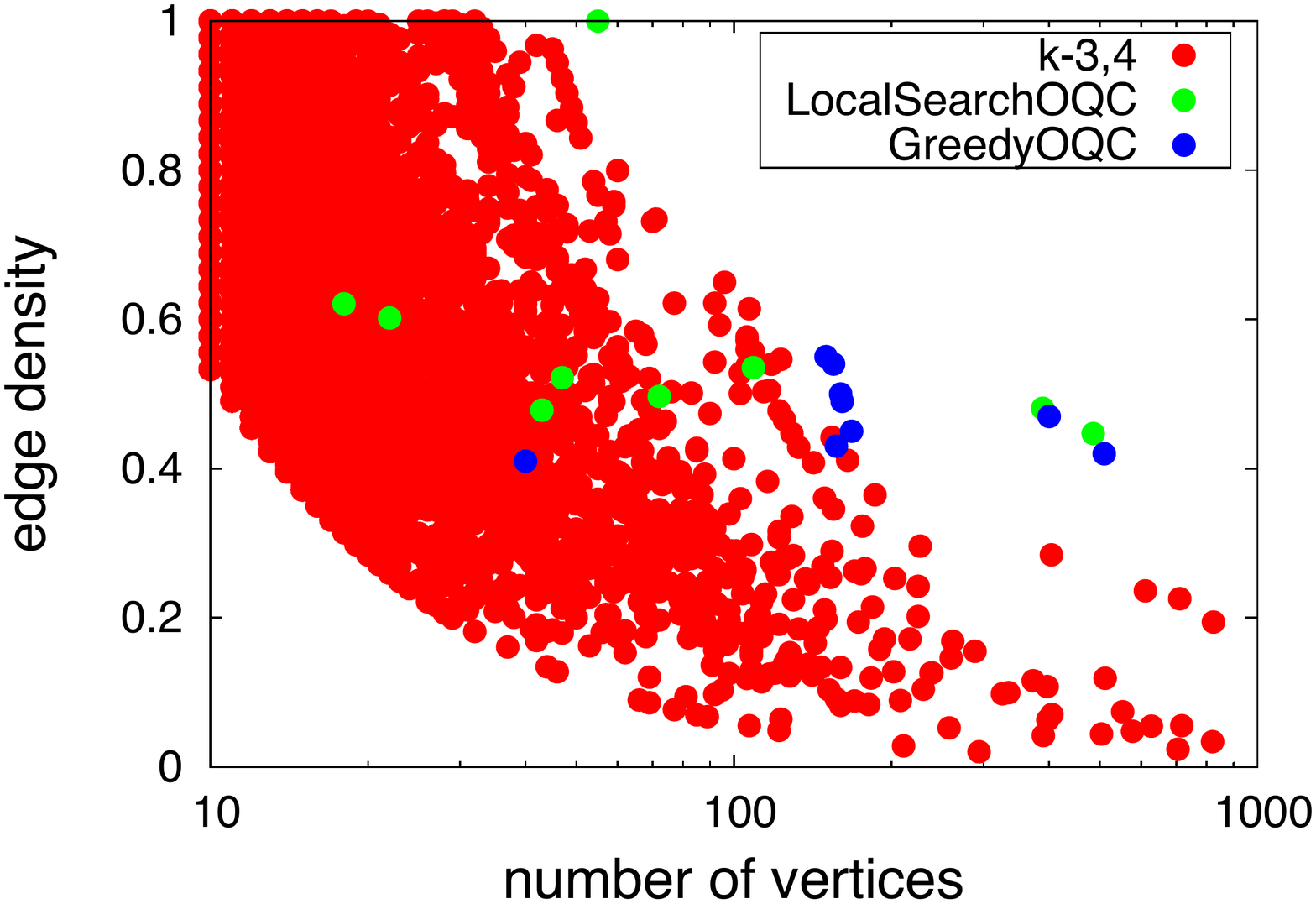}\label{fig:wiki-scatter} } \ 
   \caption{Density vs. size plots for nuclei of three graphs. State-of-the-art algorithms
   are depicted with OQC variants, and they report one subgraph at each run. We ran
   them 10 times to get a general picture of the quality. Overall, $(3,4)$-nuclei is very competitive
   with the state-of-the-art and produces many number of subgraphs with high quality and non-trivial 
   sizes.\label{fig:scatter}}
\end{figure*}

The $(3,4)$-nuclei figures  provide a useful visualization of the dense subgraph structure.
The \notredame{} has a million edges, and it is not possible to see the graph as a whole. But
the forest of nuclei breaks it down into meaningful parts, which can be visually inspected.
The overall forest is large (about $2000$ nuclei), but the nesting structure makes it easy
to absorb. We have not presented the results here, but even the \wiki{} graph of 38 million edges has
about a forest of only $4000$ nuclei (which we were able to easily visualize by a drawing tool).

Other choices of $r,s$ for the nuclei do not lead to much branching. We present
all nucleus trees for $r < s \leq 4$ for the \fb~graph in \Fig{tree1} (except $(3,4)$ which is given in 
\Fig{fb-forest}). Clearly, when $r=1$,
the nucleus decomposition is boring. For $r=2$, some structure arises, but not as dramatic
of \Fig{fb-forest}. Results vary over graphs, but for $r=1$,
there is pretty much just a chain of nuclei. For $r=2$, some graphs show more branching, but we consistently
see that for $(3,4)$-nuclei, the forest of nuclei is always branched.

\subsection{Dense subgraph discovery} \label{sec:discovery}

We plot the density histograms of the $(3,4)$-nuclei for various graphs in \Fig{density}. The $x$-axis
is (binned) density and the $y$-axis is the number of nuclei (all at least $10$ vertices) 
with that density. It can be clearly observed that we find many non-trivial dense subgraphs. It is surprising to see how many near cliques (density > 0.9)
we find. We tend to find more subgraphs of high density, and other than the \notredame{} graph,
the mass of the histogram is shifted to the right. The number of subgraphs of density at least $0.5$ 
is in the order of hundreds (and more than a thousand for \wiki).

An alternate presentation of the dense subgraphs is a scatter plot of all $(3,4)$-nuclei
with size in vertices versus density. This is given in \Fig{fb-comp} and \Fig{scatter},
where the red dots correspond to the nuclei. We see that dense subgraphs are obtained in all
scales of size, which is an extremely important feature. Nuclei capture more than just 
the densest (or high density) subgraphs, but find large sets of lower density (say around $0.2$).
Note that $0.2$ is a significant density for sets of hundreds of vertices.

\begin{figure*}[!t]
  \centering
  \subfloat[\fb]{\includegraphics[width=0.23\textwidth,keepaspectratio]{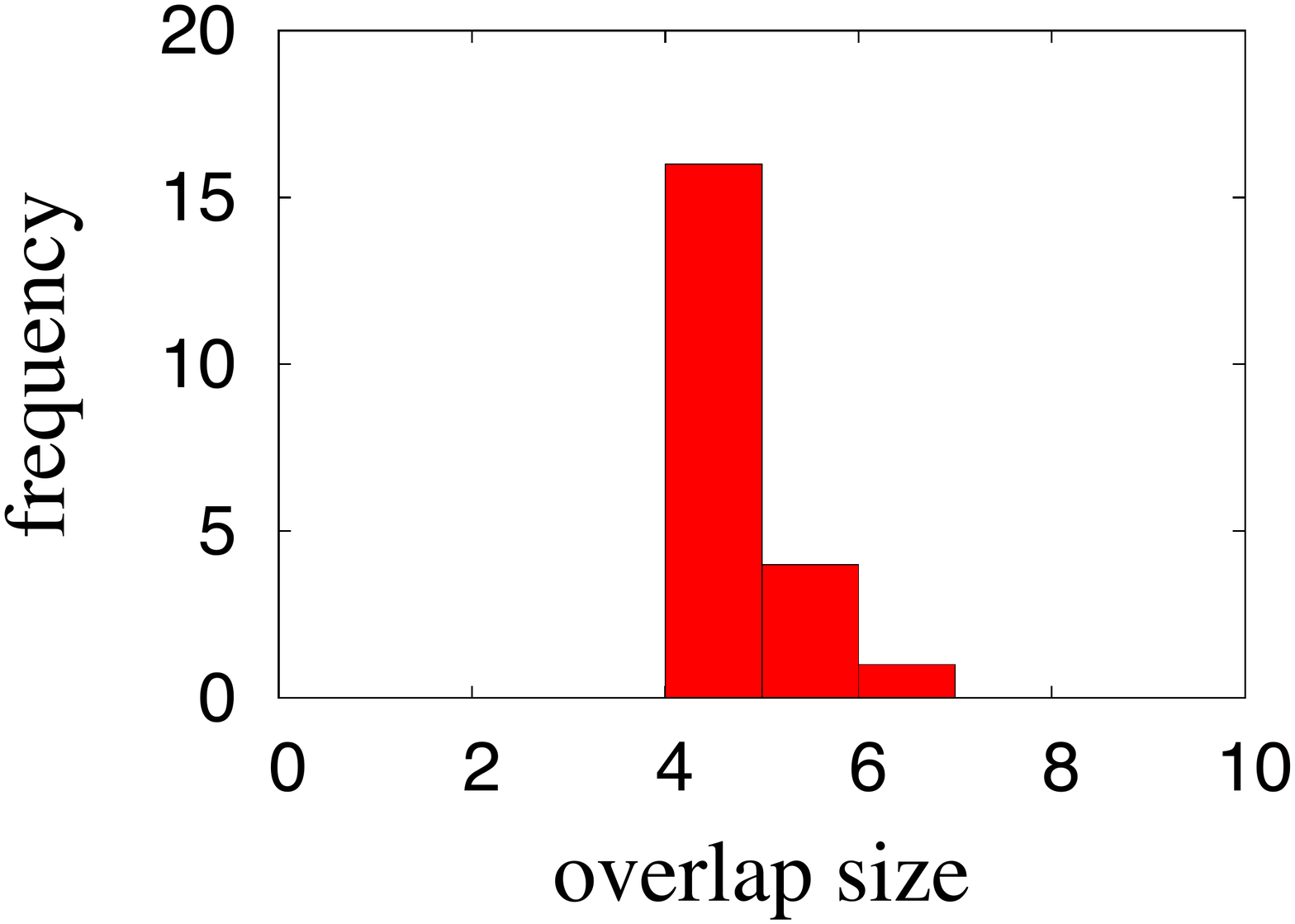}\label{fig:fb-overlap} } \ 
  \subfloat[\epinion]{\includegraphics[width=0.23\textwidth,keepaspectratio]{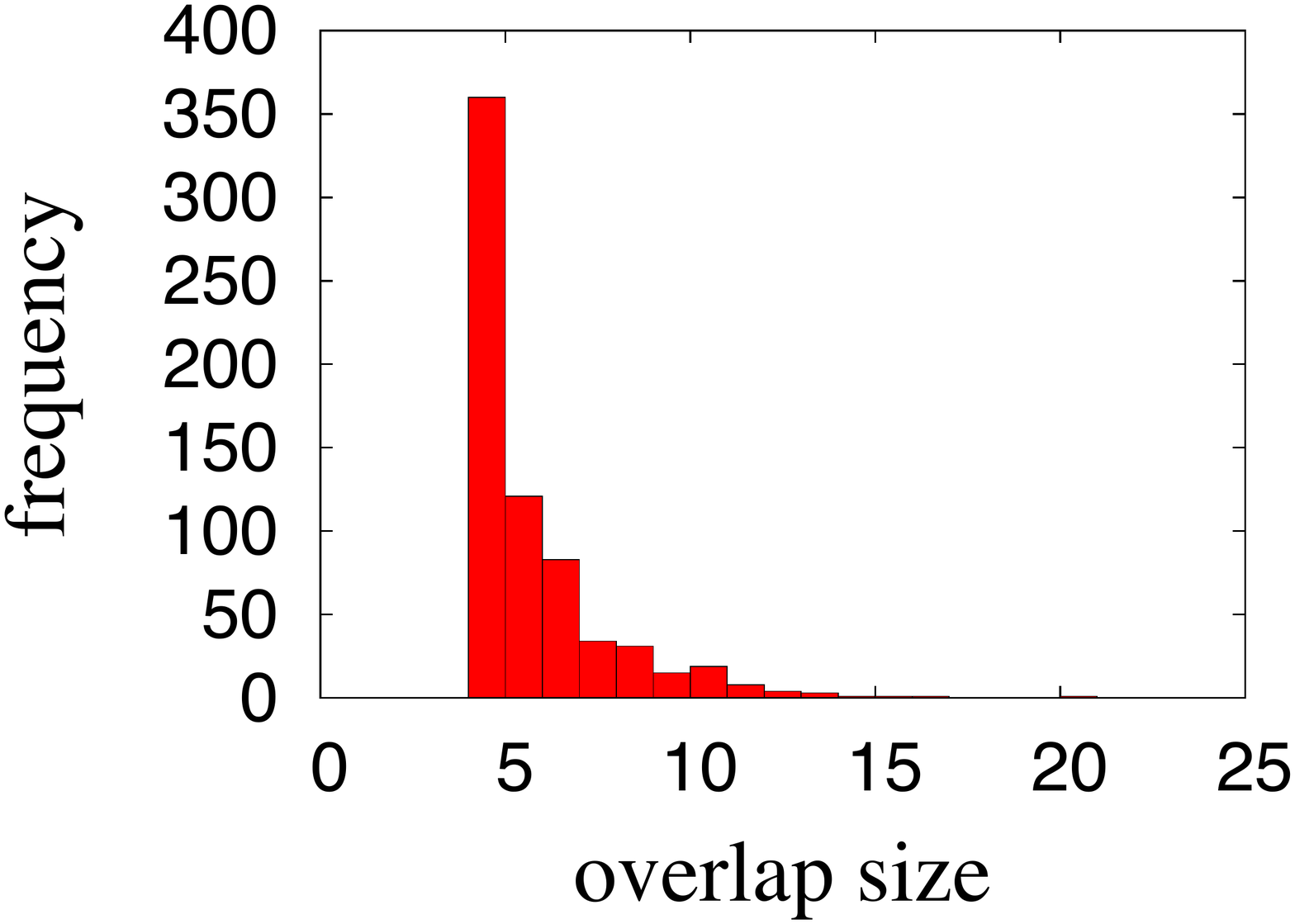}\label{fig:epinion-overlap} } \ 
  \subfloat[\notredame]{\includegraphics[width=0.23\textwidth,keepaspectratio]{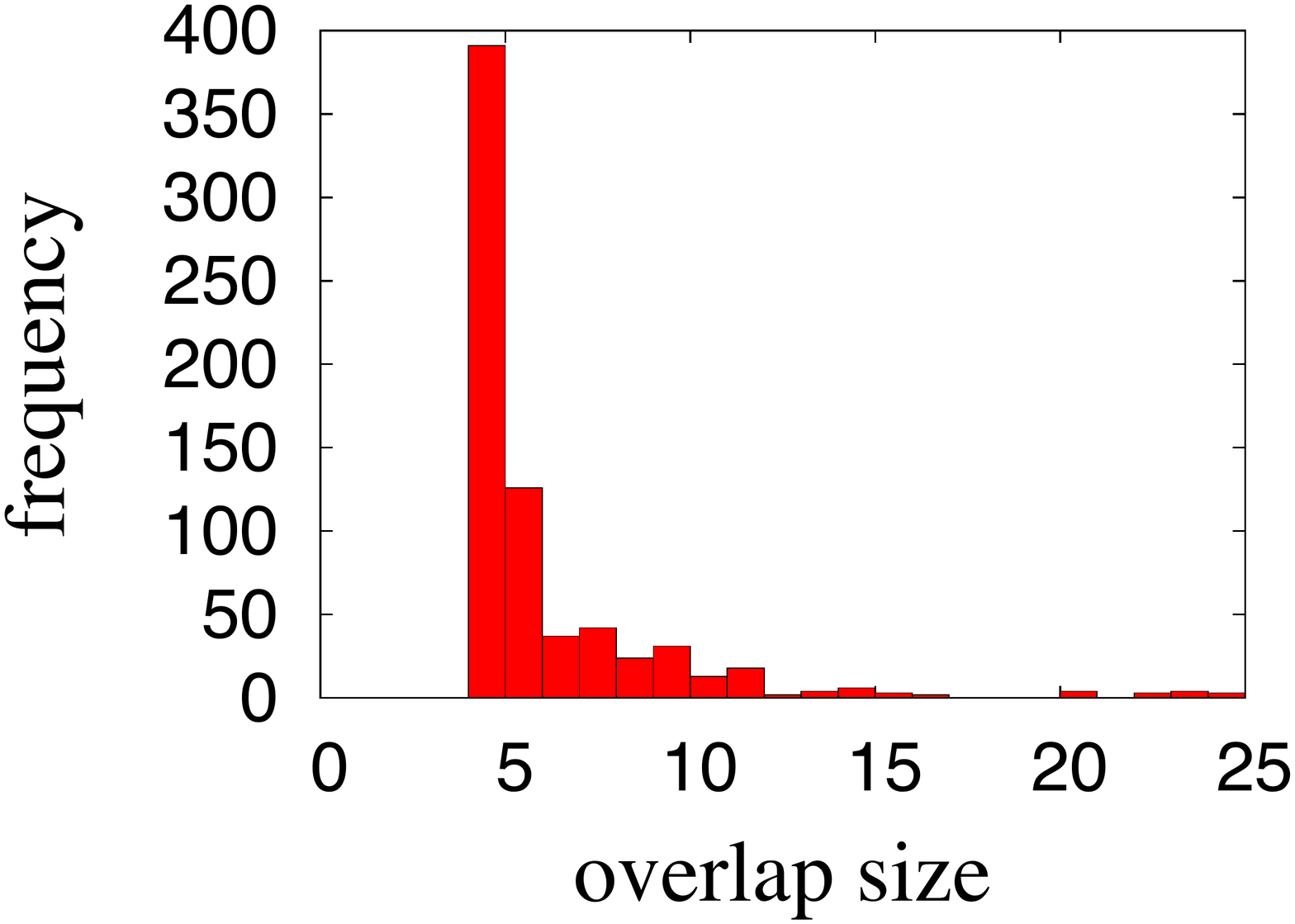}\label{fig:nd-overlap} } \ 
  \subfloat[\wiki]{\includegraphics[width=0.23\textwidth,keepaspectratio]{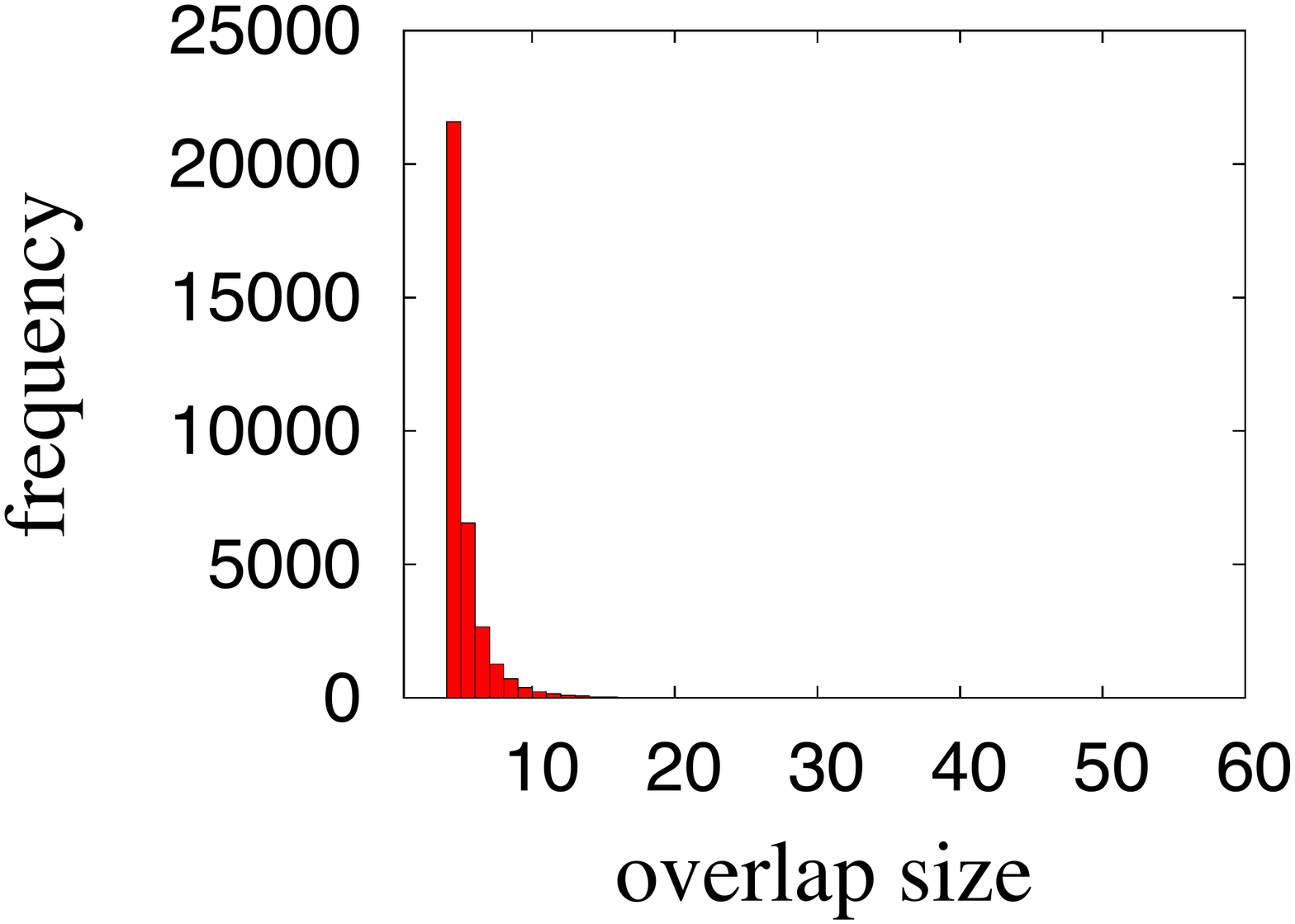}\label{fig:wiki-overlap} } \ 
   \caption{Histograms over non-trivial overlaps for $(3,4)$-nuclei. Child-ancestor intersections
   are omitted. Overlap size is in terms of the number of vertices. Most overlaps are small in size.
   We also observe that $(2,s)$-nuclei give almost no overlaps.\label{fig:overlap}}
\end{figure*}

\begin{figure*}[!t]
  \centering
  \subfloat[\fb]{\includegraphics[width=0.23\textwidth,keepaspectratio]{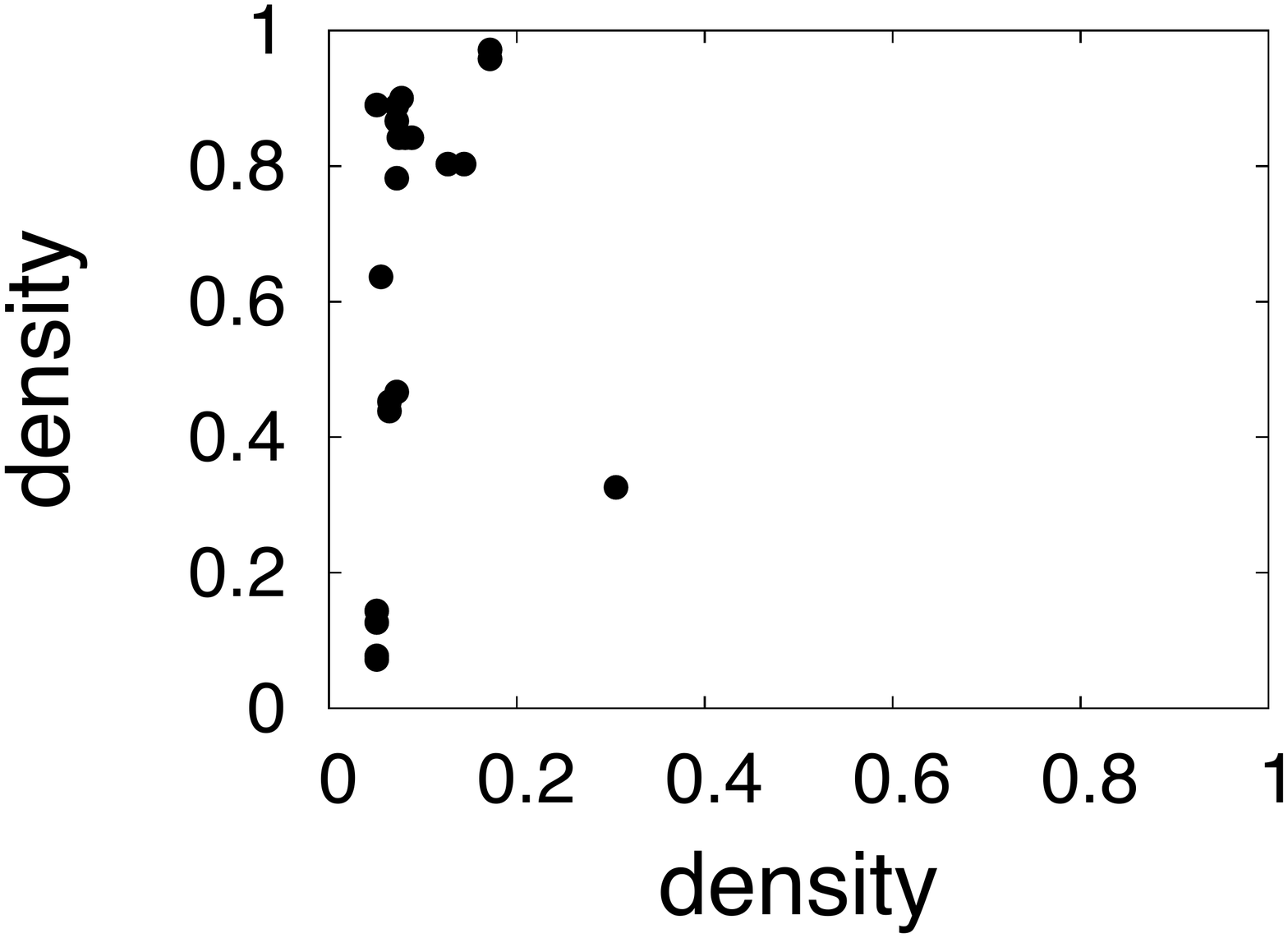}\label{fig:fb-over-scatter} } \ 
  \subfloat[\epinion]{\includegraphics[width=0.23\textwidth,keepaspectratio]{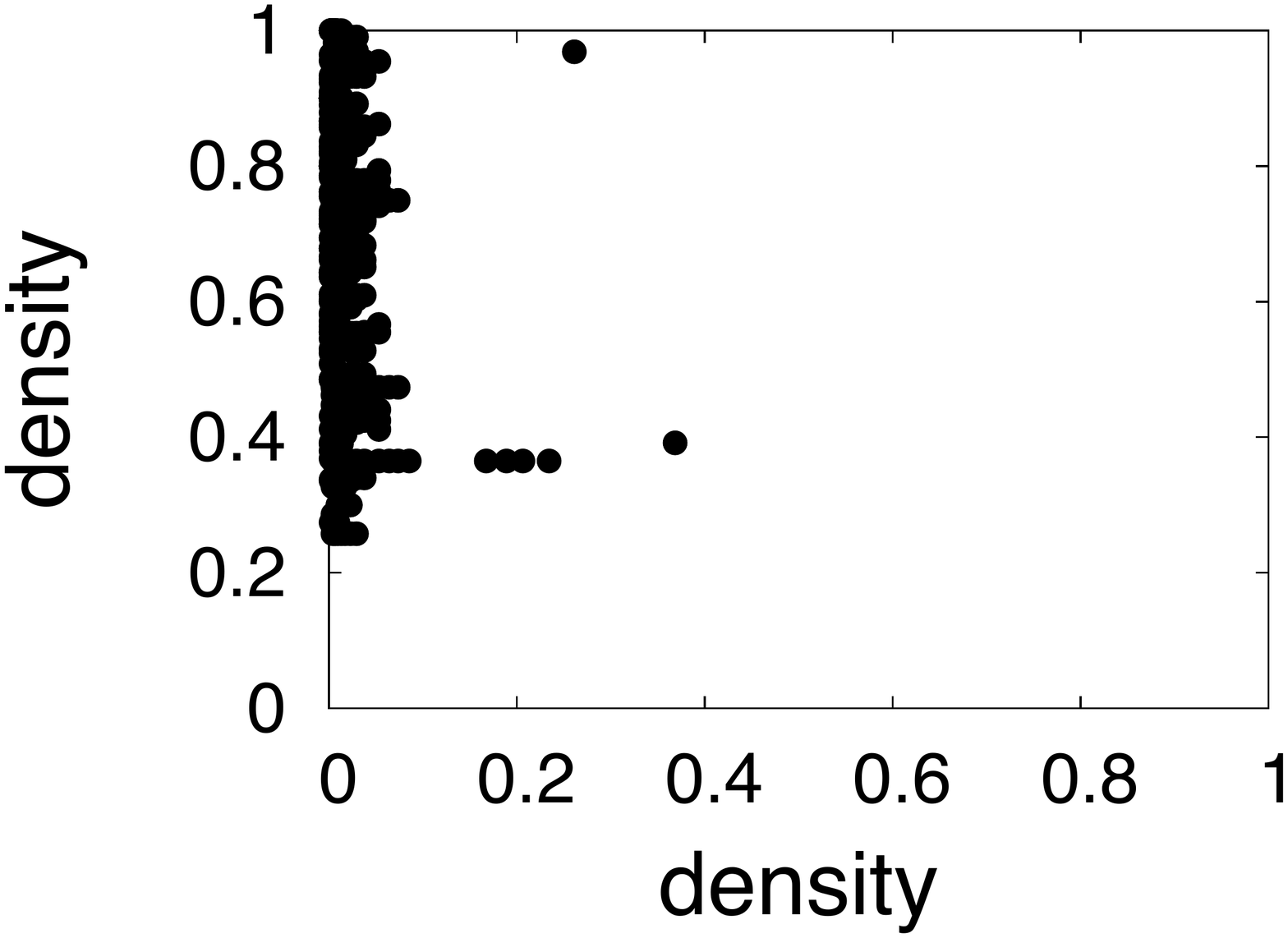}\label{fig:epinion-over-scatter} } \ 
  \subfloat[\notredame]{\includegraphics[width=0.23\textwidth,keepaspectratio]{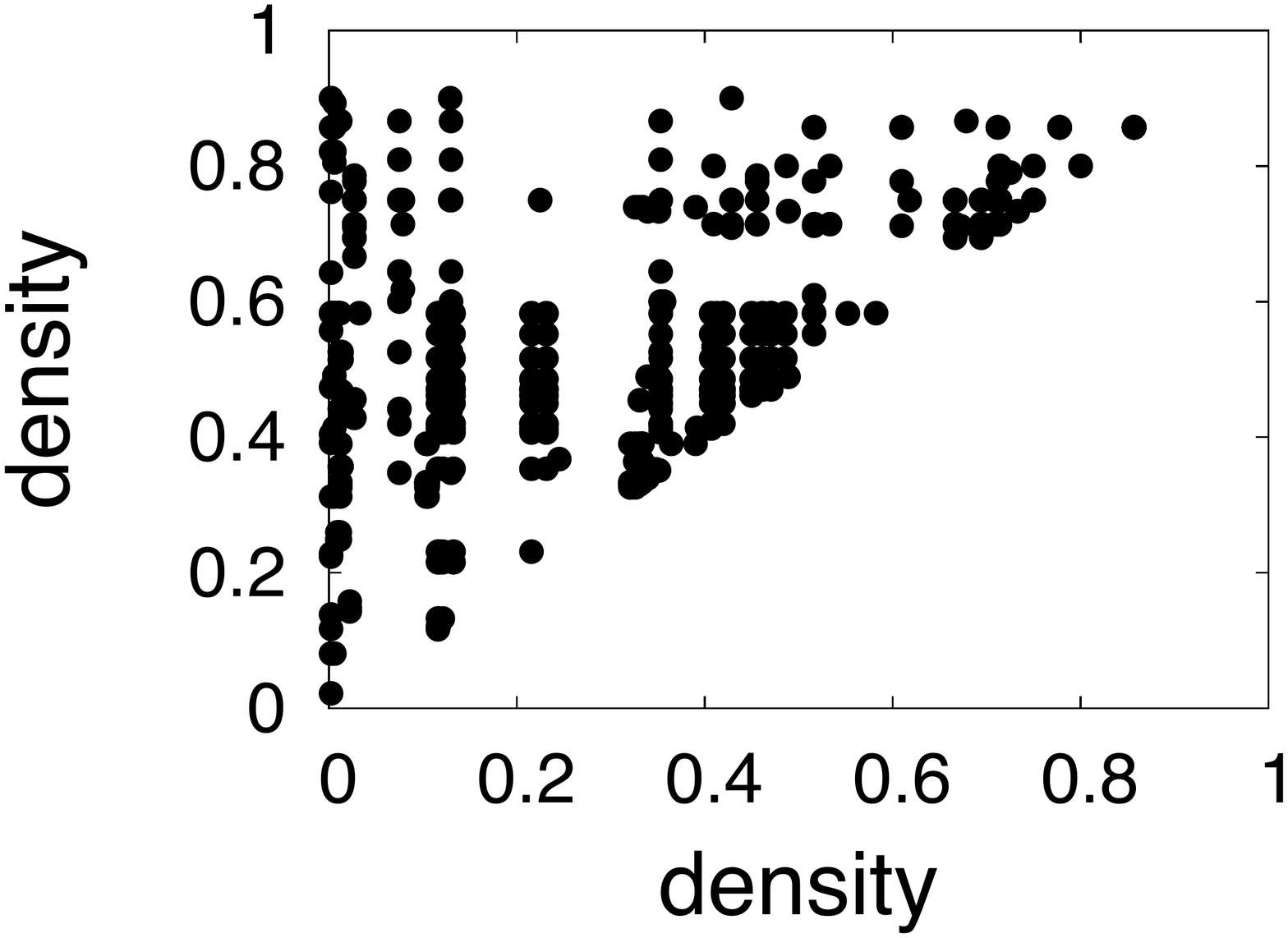}\label{fig:nd-over-scatter} } \ 
  \subfloat[\wiki]{\includegraphics[width=0.23\textwidth,keepaspectratio]{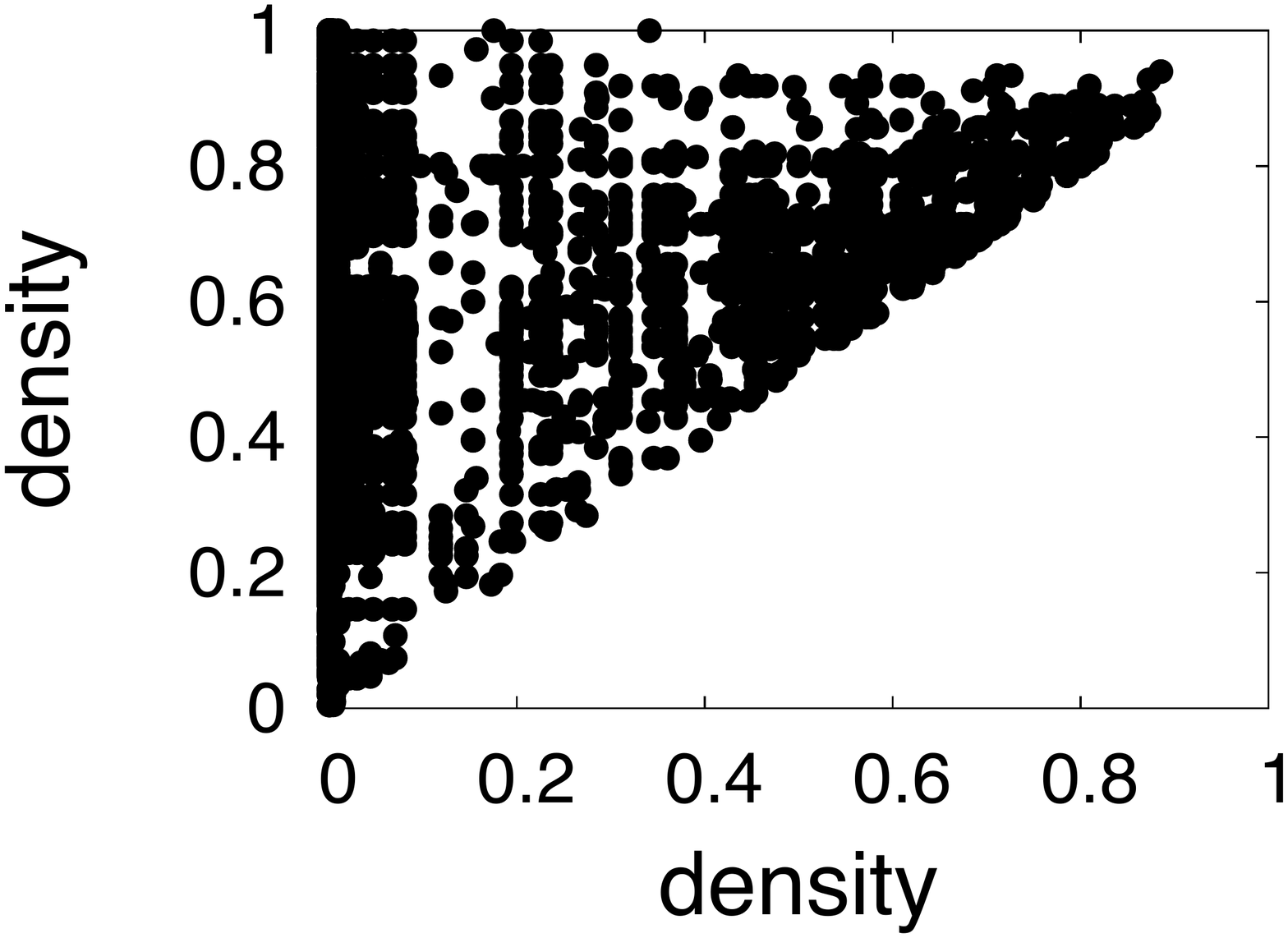}\label{fig:wiki-over-scatter} } \ 
   \caption{Overlap scatter plots for $(3,4)$-nuclei. Each axis shows the edge density of a participating
   nucleus in the pair-wise overlap. Larger density is shown on the $y$-axis. $(3,4)$-nuclei is able to get 
   overlaps between very dense subgraphs, especially in \notredame{} and \wiki. In \wiki{} graph, there are
   $1424$ instances of pair-wise overlap between two nuclei, where each nucleus has the density of at least $0.8$.
   \label{fig:over-scatter}}
\end{figure*}

\subsubsection{Comparisons with previous art} \label{sec:comp}

How does the quality of dense subgraphs found compare to the state-of-the-art?
In the scatter plots of \Fig{fb-comp} and \Fig{scatter}, we also 
show the output of two algorithms of~\cite{Tsourakakis13} in green and blue. 
The idea of~\cite{Tsourakakis13} is to approximate \emph{quasi-cliques},
and their result provides two every elegant algorithms for this process. (We collectively 
refer to them as OQC.) OQC algorithms only give a single output,
so we performed multiple runs to get many dense subgraphs. This is consistent with what
was done in~\cite{Tsourakakis13}. OQC algorithms clearly beat previous heuristics and
it is fair to say that~\cite{Tsourakakis13} is the state-of-the-art.

The $(3,4)$-nucleus decomposition does take significantly longer than the algorithms of~\cite{Tsourakakis13}.
But we always get much denser subgraphs in all runs. Moreover, the sizes are comparable
if not larger than the output of~\cite{Tsourakakis13}. Surprisingly, in \fb{} and \epinion,
some of the best outputs of OQC are very close to $(3,4)$-nuclei.
Arguably, the $(3,4)$-nuclei perform worst on \wiki, where OQC
find some larger and denser instances than $(3,4)$-nuclei. Nonetheless, the smaller $(3,4)$-nuclei
are significantly denser. We almost always can find fairly large cliques.

In \Tab{properties}, we consider the OQC output vs $(3,4)$-nuclei for all graphs.
Barring 4 instances, there is a $(3,4)$-nucleus that is larger and denser
than the OQC output. In all cases but one (adjnoun), there is a $(3,4)$-nucleus of density (of non-trivial size) higher than the 
the OQC output. The nuclei have the advantage of being the output of a fixed, deterministic procedure, and not
a heuristic that may give different outputs on different runs. We mention that OQC algorithms have a significant running time advantage
over finding $(3,4)$-nuclei, for a single subgraph finding.

\subsection{Overlapping nuclei} \label{sec:overlap}

A critical aspect of nuclei is that they can overlap. Grappling with overlap
is a major challenge when dealing with graph decompositions. We believe
one of the benefits of nuclei is that they naturally allow for (restricted) overlap.
As mentioned earlier, no two $(r,s)$-nuclei can contain the same $K_r$.
This is a significant benefit of setting $r=3,s=4$ over other choices.

In \Fig{overlap}, we plot the histogram over non-trivial overlaps for $(3,4)$-nuclei. 
(We naturally do not consider a child nucleus intersecting with an ancestor.)
For a given overlap size in vertices, the frequency is the number of pairs
of $(3,4)$-nuclei with that overlap. This is shown
for four different graphs. The total number of pair-wise overlaps (the sum of frequencies)
is typically around half the total number of $(3,4)$-nuclei. We observed
that the Jaccard similarities are less than $0.1$ (usually smaller).
This suggest that we have large nuclei with some overlap. 

There are bioinformatics applications for finding vertices that are present in numerous dense subgraphs~
\cite{Hu05}.
The $(3,4)$-nuclei provide many such vertices. In \Fig{over-scatter}, we give a scatter plot of all intersecting
nuclei, where nuclei are indexed by density. For two intersecting nuclei of density 
$\alpha > \beta$, we put a point $(\alpha,\beta)$. We only plot pairs where the overlap
is at least $5$ vertices. Especially for \notredame{} and \wiki, we get significant overlaps between dense 
clusters.

In contrast, for all other settings of $r,s$, we get almost no overlap. When $r = 2$,
nuclei can only overlap at vertices, and this is too stringent to allow for interesting
overlap.\\\\

\subsection{Runtime results}\label{sec:runtime}

\Tab{properties} presents the runtimes in seconds for the entire construction.
To provide some context, we describe runtimes for varying choices of $r,s$.
For $r=1, s=2$ ($k$-cores), the decomposition is linear and extremely fast.
For the largest graph (\wiki) we have, with $39M$ edges,
it takes only $4.26$ seconds. For $r=2, s=3$ (trusses), the time can be two orders
of magnitude higher. And for $(3,4)$-nuclei, it is an additional order of magnitude higher.
Nonetheless, our most expensive run took less than an hour on the \wiki{} graph,
and the final decomposition is quite insightful. It provides about 6000 nuclei with more
than 10 vertices, most of them of have density of at least $0.4$. The algorithms of~\cite{Tsourakakis13} take 
roughly a minute for \wiki{} to produce \emph{only one} dense subgraph.

The theoretical running time analysis of \Thm{slow} gives a running time bound of $\sum_v c_3(v) d(v)$.
In \Tab{properties}, we show this value for the various graphs. In general, we note that
this value roughly correlates with the running time. For graphs where the running time is in many minutes,
this quantity is always in the billions. For the large wiki graphs where the $(3,4)$-nucleus decomposition
is most expensive, this is in the trillions.

\section{Further directions}\label{sec:future}

The most important direction is in the applications of nucleus decompositions.
We are currently investigating bioinformatics applications, specifically protein-protein and
protein-gene interaction networks. Biologists often want a global view of the dense substructures, and we believe
the $(3,4)$-nuclei could be extremely useful here. In our preliminary analyses, we wish to see if the nuclei
pick out specific functional units. If so, that would provide strong validation of dense subgraph analyses
for bioinformatics.

It is natural to try even larger values of $r, s$. Preliminary experimentation suggested that this gave little
benefit in either the forest or the density of nuclei. Also, the cost of clique enumeration becomes
forbiddingly large. It would be nice to argue that $r=3, s=4$ is a sort of sweet spot for nucleus decompositions.
Previous theoretical work suggests that any graph with a sufficient triangle count undergoes special 
``community-like" decompositions~\cite{GuRoSe14}. That might provide evidence to why triangle based nuclei
are enough.

A faster algorithm for the $(3,4)$-nuclei is desirable. Clique enumeration is a well-studied problem~\cite{Bron73},
and we hope techniques from these results may provide ideas here. Of course, as we said earlier, any method
based on storing 4-cliques is infeasible (space-wise). We hope to devise a clever algorithm or data structure that quickly
determines the 4-cliques that a triangle participates in.

Last but not least, we seek for incremental algorithms to maintain the $(r,s)$-nuclei for a stream of edges.
There are existing techniques for streaming $k$-core algorithms~\cite{Sariyuce13-VLDB} and we believe that similar 
methods can be adapted for $(r,s)$-nuclei maintenance.

\section{Acknowledgements}

We are grateful to Charalampos Tsourakakis for sharing his code base for~\cite{Tsourakakis13}.
This work was funded by the DARPA GRAPHS program. Sandia National Laboratories is a multi-program laboratory managed and operated by Sandia Corporation, 
a wholly owned subsidiary of Lockheed Martin Corporation, for the  U.S. Department of Energy's National Nuclear Security  Administration under contract DE-AC04-94AL85000.

\bibliographystyle{abbrv}
\bibliography{paper,sesh}  

\end{document}